\documentclass[10pt]{article}
\usepackage{amsfonts}
\usepackage{amsfonts}
\usepackage{amsfonts}
\usepackage{amsfonts}
\usepackage{amsfonts}
\usepackage{amsfonts}
\usepackage{amsfonts}
\usepackage{amsfonts}
\usepackage{amsfonts}
\usepackage{amsfonts}
\usepackage{amsfonts}
\usepackage{amsfonts}
\usepackage{amsfonts}
\usepackage{amsfonts}
\usepackage{amsfonts}
\usepackage{amsfonts}
\usepackage{amsfonts}
\usepackage{amsfonts}
\usepackage{amsfonts}
\usepackage{amsfonts}
\usepackage{amsfonts}
\usepackage{amsfonts}
\usepackage{amsfonts}
\usepackage{amsfonts}
\usepackage{amsfonts}
\usepackage{amsfonts}
\usepackage{amsfonts}
\usepackage{amsfonts}
\usepackage{amsfonts}
\usepackage{amsfonts}
\usepackage{amsfonts}
\usepackage{amsfonts}
\usepackage{amsfonts}
\usepackage{amsfonts}
\usepackage{amsfonts}
\usepackage{amsfonts}
\usepackage{amsfonts}
\usepackage{amsfonts}
\usepackage{amsfonts}
\usepackage{amsfonts}
\usepackage{amsfonts}
\usepackage{amsfonts}
\usepackage{amsfonts}
\usepackage{amsfonts}
\usepackage{amsfonts}
\usepackage{amsfonts}
\usepackage{amsfonts}
\usepackage{amsfonts}
\usepackage{amsfonts}
\usepackage{amsfonts}
\usepackage{amsfonts}
\usepackage{amsfonts}
\usepackage{amsfonts}
\usepackage{amsfonts}
\usepackage{amsfonts}
\usepackage{amsfonts}
\usepackage{amsfonts}
\usepackage{amsfonts}
\usepackage{amsmath}
\usepackage{amssymb, amscd, amsthm}
\usepackage[all]{xy}
\usepackage[dvips]{graphicx}
\usepackage{verbatim}
\usepackage{cite}
\usepackage{enumerate}
\usepackage{float}
\usepackage{diagbox}

\numberwithin{equation}{section}
\newtheorem{lemma}{\textbf{Lemma}}[section]
\newtheorem{theorem}{\textbf{Theorem}}
\newtheorem{remark}{\textbf{Remark}}[section]
\newtheorem{corollary}{\textbf{Corollary}}[section]

\newtheorem{proposition}{\textbf{Proposition}}[section]

\newtheorem{Definition}{\textbf{Definition}}
\usepackage{longtable}

\usepackage{multirow}
\usepackage{tikz}

\begin{document}

\baselineskip 17pt\title{\Large\bf Codes with Weighted Poset Block Metrics}
\author{\large  Wen Ma \quad\quad Jinquan Luo\footnote{The authors are with School of Mathematics
and Statistics \& Hubei Key Laboratory of Mathematical Sciences, Central China Normal University, Wuhan China.\newline
 E-mails: mawen95@126.com(W.Ma),  luojinquan@mail.ccnu.edu.cn(J.Luo)}}
\date{}
\maketitle
\date{}
\maketitle

{\bf Abstract}: Weighted poset block metric is a generalization of weighted poset metric introduced by Panek et al. ([\ref{panek}]) and the metric for linear error-block codes introduced by Feng et al. ([\ref{FENG}]). This type of metrics includes many classical metrics such as Hamming metric, Lee metric, poset metric, pomset metric, poset block metric, pomset block metric and so on. In this work, we focus on constructing new codes under weighted poset block metric from given ones. Some basic properties such as minimum distance and covering radius are determined.


\section{Introduction}

\quad\; Let $\mathbb{F}_q^n$ be the spaces of $n$-tuples over a finite field $\mathbb{F}_q$. Most of the coding theory was developed considering the metric determined by Hamming weight on $\mathbb{F}_q^n$. The study of codes endowed with a metric other than the Hamming metric gained momentum since 1990's. In 1995, Brualdi, Graves and Lawrence introduced \emph{poset metric}, which is defined by partial orders on the set of coordinate positions of $\mathbb{F}_q^n$  ([\ref{BG}]). Poset metric is a generalization of the Hamming metric, in the sense that the latter is attained by considering the trivial order. This has been a fruitful approach, since a number of unusual properties arise in this context such as intriguing relative abundance of MDS and perfect codes ([\ref{HYUN}], [\ref{panek}]).  Over the last two decades, the study of codes in the poset metric has made many developments in different subjects in coding theory.

Feng, Xu and Hickernell ([\ref{FENG}]) introduced the block metric by partitioning the set of coordinate positions of $\mathbb{F}_q^n$ and studied MDS block codes. In 2008, Alves, Panek and Firer combined the poset and block structure, obtaining a further generalization called the \emph{poset block metrics} ([\ref{ALVES}]). A particular instance of poset block codes and spaces, with one-dimensional blocks, are the spaces introduced by Niederreiter in 1991 ([\ref{NIED}]) and Rosenbloom and Tafasman in 1997 ([\ref{ROSEN}]). Later, Dass, Sharma and Verma obtained a Singleton type bound for poset block codes([\ref{DASS}]). A code meeting this bound is called a maximum distance separable poset block code. \emph{Niederreiter-Rosenbloom-Tsfasman block metric} (in short, \emph{NRT block metric}) is a particular case of poset block metric when the poset is a chain ([\ref{panek}]).

As the support of a vector $v$ in $\mathbb{F}_q^n$ is a set which induces order ideals and metrics on $\mathbb{F}_q^n$, the poset metric codes could not accommodate Lee metric structure due to the fact that the support of a vector with respect to Lee weight is not a set but rather a multiset. In order to handle Lee metric, a much general class of metrics called \emph{pomset metric} is introduced by Irrinki and Selvaraj ([\ref{POMSET}],[\ref{POMSETMAC}],[\ref{POMSETMDS}]) for codes over $\mathbb{Z}_m$. Over $\mathbb{Z}_2^n$ and $\mathbb{Z}_3^n$ the pomset metric is actually the poset metric. Moreover, when the pomset is induced by an antichain, it is the Lee metric.

More recently in [\ref{PANEK}] and [\ref{ANTI}], Panek and Pinheiro has proposed and studied weighted poset weight for finite field alphabet, which is a generalization of both the poset metric and pomset metric. In particular, it is a generalization of Hamming metric and Lee metric.

In this work, we combine weighted poset metric with error-block metric to obtain a further generalization called the\emph{ weighted poset block metric} which includes not only all additive metrics mentioned above but also some block metric such as \emph{poset block metric}, \emph{block metric}, \emph{pomset block metric} and so on.

It is known that many interesting and important codes will arise by modifying or combining existing codes under classical Hamming metric ([\ref{HUFFMAN}]).  There are also several different ways to join two ordered sets together ([\ref{Davey}]). The poset structure that could be imposed on the resultant codes will have its effect on the minimum distance and covering radius.

The remainder of the paper is organized as follows. In Section 2, we give some definitions, notations and basic facts of posets and weighted poset block weight over $\mathbb{F}_q^n$. In section 3, we consider the packing radius and covering radius of a $(P,\pi,w)$-code when $P$ is a chain. When $w$ is taken to be the Hamming weight, our conclusion will coincide with the result under NRT block metric. In section 4, we give several different ways to construct new $(P,\pi,w)$-codes from given ones. We introduce the concept of the direct sum and direct product of the labeling maps. The new poset block structure that could be imposed on the resultant codes. We focus on discussing its effect on minimum distance and covering radius.

\section{Preliminary}

\quad\; In the following, we give some basic definitions and notations about poset that are used throughout the remainder of the paper. For more details of posets see [\ref{Davey}].

Let $P$ be a set. A \emph{partial order} on $P$ is a binary relation $\leq_P$ on $P$ such that for all $x,y,z\in P$, we have $x\leq_P x$ (\emph{reflexivity}), $x\leq_P y$ and $y\leq_P x$ imply $x=y$ (\emph{antisymmetry}), $x\leq_P y$ and $y\leq_P z$ imply $x\leq_P z$ (\emph{transitivity}). A set $P$ equipped with an order relation $\leq_P$ is said to be a \emph{poset}. A poset $P$ is a \emph{chain} if any two elements of $P$ are comparable. The opposition of a chain is an antichain, that is, poset $P$ is an \emph{antichain} if $x\leq y$ in $P$ only when $x=y$. We call a subset $Q$ of $P$ an \emph{ideal} if, whenever $x\in Q$, $y\in P$ and $y\leq_P x$, we have $y\in Q$. For a subset $E$ of $P$, the \emph{ideal generated by $E$}, denoted by $\langle E\rangle_P$, is the smallest ideal of $P$ containing $E$. We prefer to denote the ideal generated by $\{i\}$ as $\langle i\rangle$ instead of $\langle\{i\}\rangle$. We denoted by $\langle i\rangle^{*}$ the \emph{difference} $\langle i\rangle-\{i\}=\{j\in P: j<i\}$.

There are several different ways to construct a new poset from any two given posets.
Let $P$ and $Q$ be two posets.
\begin{itemize}
\item \textbf{Disjoint union:} The disjoint union of $P$ and $Q$ denoted by $P\uplus Q$ is the poset formed by defining order relation on the underlying set $P\cup Q$:
  \begin{center}
  $x\leq y$ in $P\cup Q$ $\Leftrightarrow$ ($x,y\in P$ and $x\leq_P y$) or ($x,y\in Q$ and $x\leq_Q y$).
  \end{center}
\item \textbf{Linear sum:} The linear sum of $P$ and $Q$ denoted by $P\oplus Q$ is also a poset whose order relation is defined on $P\cup Q$ in the following way:
  \begin{center}
  $x\leq y$ in $P\cup Q$ $\Leftrightarrow$ ($x,y\in P$ and $x\leq_P y$) or ($x,y\in Q$ and $x\leq_Q y$) or ($x\in P$ and $y\in Q$).
  \end{center}
\item \textbf{Cartesian product:} Denote by $P\times Q=\{(i,j):i\in P, j\in Q\}$. Define an order relation $\leq$ on the underlying set $P\times Q$ as
  $$(x,y)\leq(x',y')\ \text{in}\ P\times Q\Leftrightarrow x\leq_P x'\ \text{and}\ y\leq_Q y'.$$
  Then $P\times Q$ is a poset with the order relation defined above and is called Cartesian product of $P$ and $Q$, denoted by $P\otimes Q$.
\item \textbf{Lexicographic product:} Define an order relation $\leq$ on the underlying set $P\times Q$ again in a different way:
  \begin{center}
  $(x,y)\leq(x',y')$ in $P\times Q$ $\Leftrightarrow$ ($x<_P x'$) or ($x=x'$ and $y\leq y'$).
  \end{center}
  Then $P\times Q$ is a poset with this order relation and is called lexicographic product of $P$ and $Q$, denoted by $P\star Q$.
\end{itemize}

\quad\;In the following, we give two ways to obtain a new poset from the old one.
Let $P$ be a poset.
\begin{itemize}
\item \textbf{Puncturing:} We can get a new poset $P^{-}$ from $P$ by deleting an element $z\in P$ and its order relation is defined as:
    $$x\leq y\in P^{-}\Leftrightarrow x\leq y\in P.$$
\item \textbf{Extending:} By adding an element $z$ in $P$, we obtain a new poset $P^{+}$ whose order relation is defined as:
    \begin{center}
    $x\leq y\in P^{+}$ $\Leftrightarrow$ ($x\leq y$ in $P$) or ($x=y=z$).
    \end{center}
\end{itemize}

\begin{remark}
By the definitions of puncturing poset and extending poset, we get the following.
\begin{enumerate}[(1)]
\item When poset $P$ is a chain (antichain), then $P^{-}$ is a chain (antichain).
\item Poset $P^{+}$ can never be a chain.
\end{enumerate}

\end{remark}

The definitions of weight and metric can be defined on general rings. In particular, we restrict it to finite field because it is the most explored topic in the context of coding theory.

Let $\mathbb{F}_q$ be the finite field of order $q$ and $\mathbb{F}_q^n$ the $n$-dimensional vector space over $\mathbb{F}_q$.

\begin{Definition}
A map $d: \mathbb{F}_q^n\times \mathbb{F}_q^n\rightarrow \mathbb{N}$ is a metric on $\mathbb{F}_q^n$ if it satisfies the following conditions:
\begin{enumerate}[(1)]
\item (non-negativity) $d(\boldsymbol{u,v})\geq 0$ for all $\boldsymbol{u,v}\in \mathbb{F}_q^n$ and $d(\boldsymbol{u,v})=0$ if and only if $\boldsymbol{u}=\boldsymbol{v}$;
\item (symmetry) $d(\boldsymbol{u,v})=d(\boldsymbol{v,u})$ for all $\boldsymbol{u,v}\in \mathbb{F}_q^n$;
\item (triangle inequality) $d(\boldsymbol{u,v})\leq d(\boldsymbol{u,w})+d(\boldsymbol{w,v})$ for all $\boldsymbol{u,v,w}\in \mathbb{F}_q^n$.
\end{enumerate}
\end{Definition}

\begin{Definition}
A map $w: \mathbb{F}_q^n\rightarrow \mathbb{N}$ is a weight on $\mathbb{F}_q^n$ if it satisfies the following conditions:
\begin{enumerate}[(1)]
\item $w(\boldsymbol{u})\geq 0$ for all $\boldsymbol{u}\in \mathbb{F}_q^n$ and $w(\boldsymbol{u})=0$ if and only if $\boldsymbol{u}=\boldsymbol{0}$;
\item $w(\boldsymbol{u})=w(\boldsymbol{-u})$ for all $\boldsymbol{u}\in \mathbb{F}_q^n$;
\item $w(\boldsymbol{u+v})\leq w(\boldsymbol{u})+w(\boldsymbol{v})$.
\end{enumerate}
\end{Definition}

It is straightforward to prove that, if $w$ is a weight over $\mathbb{F}_q^n$, then the map $d_w$ defined by $d(\boldsymbol{u},\boldsymbol{v})=w(\boldsymbol{u}-\boldsymbol{v})$ is a metric on $\mathbb{F}_q^n$. See [\ref{DEZA}] and [\ref{GAB}] for detailed discussion on weight and metric.

Let $\pi: [s]\rightarrow\mathbb{N}$ be a map such that $n=\sum\limits_{i=1}^s\pi(i)$. The map $\pi$ is said to be a \emph{labeling} of the poset $P$, and the pair $(P,\pi)$ is called a \emph{poset block structure} over $[s]$. Denote $\pi(i)$ by $k_i$. We take $V_i$ as the $\mathbb{F}_q$-vector space $\mathbb{F}_q^{k_i}$ for all $1\leq i\leq s$. We define $V$ as the direct sum
$$V=V_1\oplus V_2\oplus\cdots\oplus V_s$$
which is isomorphic to $\mathbb{F}_q^n$. Each $\boldsymbol{u}\in V$ can be uniquely decomposed as
$$\boldsymbol{u}=\boldsymbol{u_1}+\boldsymbol{u_2}+\cdots+\boldsymbol{u_s}$$
 where $\boldsymbol{u_i}=(u_{i_1},\ldots,u_{ik_i})\in V_i$ for $1\leq i\leq s$.

Let $w$ be a weight on $\mathbb{F}_q$ and let $P$ be a poset. Given $\boldsymbol{u}\in V$, set
\begin{enumerate}[(1)]
\item $W_i^P(\boldsymbol{u})=\max\left\{w(u_{ij}): 1\leq j\leq k_i\right\}$ for $1\leq i\leq s$;
\item $M_w=\max\left\{w(\alpha): \alpha\in\mathbb{F}_q\right\}$;
\item $m_w=\min\left\{w(\alpha): 0\neq\alpha\in\mathbb{F}_q\right\}$.
\end{enumerate}

The \emph{block support} or \emph{$\pi$-support} of $\boldsymbol{u}\in V$ is the set
$$supp_{\pi}^P(\boldsymbol{u})=\left\{i\in [s]:\boldsymbol{u_i}\neq\boldsymbol{ 0}\right\}.$$
We denote by $I_{\boldsymbol{u}}^P$ the ideal generated by $supp_{\pi}^P(\boldsymbol{u})$ and denote by $M_{\boldsymbol{u}}^P$ the set of maximal elements in the ideal $I_{\boldsymbol{u}}^P$. The \emph{$\left(P,\pi,w\right)$-weight} of $\boldsymbol{u}$ is defined as
$$\overline{w}_{w,(P,\pi)}(\boldsymbol{u})=\sum\limits_{i\in M_{\boldsymbol{u}}^P} W_i(\boldsymbol{u})+\sum\limits_{i\in I_{\boldsymbol{u}}^P\setminus M_{\boldsymbol{u}}^P}M_w.$$
For $\boldsymbol{u},\boldsymbol{v}\in V$, define their \emph{$\left(P,\pi,w\right)$-distance} as
$$d_{w,(P,\pi)}(\boldsymbol{u},\boldsymbol{v}) =\overline{\omega}_{w,(P,\pi)}(\boldsymbol{u}-\boldsymbol{v}).$$
The $(P,\pi,w)$-weight $\overline{\omega}_{w,(P,\pi)}$ and the $(P,\pi,w)$ distance $d_{w,(P,\pi)}$ is also called \emph{weighted poset block weight} and \emph{weighted poset block distance}.

The pair $\left(V,d_{w,(P,\pi)}\right)$ is said to be a \emph{weighted poset block space}. A $(P,\pi,w)$-code $C$ of length $n$ over $\mathbb{F}_q$ is a subset of $(\left(V,d_{w,(P,\pi)}\right)$.  A \emph{linear $(P,\pi,w)$-code} is a subspace of $V$. The minimum $(P,\pi,w)$-distance of a code $C$ is
$$d_{w,(P,\pi)}(C)=\min\left\{d_{w,(P,\pi)}(\boldsymbol{u},\boldsymbol{v}): \boldsymbol{u}\neq \boldsymbol{v}\in C\right\}$$
is the minimal $(P,\pi,w)$-distance of $C$.
When the weight $w$ over $\mathbb{F}_q$ is considered to be the Hamming weight $w_H$, we denote by $d_{(P,\pi)}(C)=d_{w_H,(P,\pi)}(C)$. A \emph{linear $(P,\pi,w)$-code} is a subspace of $V$.

\begin{remark}
It is worth noting that this metric combines and extends several classical metrics in coding theory. For instance,
\begin{enumerate}[(1)]
\item When the weight $w$ over $\mathbb{F}_q$ is the Hamming weight, the $(P,\pi,w)$-weight reduces to the poset block weight introduced by Alves et al. (see [\ref{ALVES}]).
\item Similarly, the Niederreiter-Rosenbloom-Tafasman block weight (NRT block weight), introduced by Panek (see [\ref{panek}]), becomes a particular case of $(P,\pi,w)$-weight when $P$ is taken to be a chain and $w$ is taken to be the Hamming weight.
\item When the label $\pi$ satisfies $\pi(i)=1$ for all $i\in[s]$, the $(P,\pi,w)$-weight reduces to the weighted poset weight introduced by Panek et al. (see [\ref{PANEK}] and [\ref{ANTI}]).
\item In case both conditions occur ($\pi(i)=1$ for all $i\in[s]$, $w$ is the Hamming weight and $P$ is the antichain order), the $(P,\pi,w)$-weight reduces to the usual Hamming weight.
\item In case both conditions occur ($\pi(i)=1$ for all $i\in[s]$, $w$ is the Lee weight and $P$ is the antichain order), the $(P,\pi,w)$-weight reduces to the usual Lee weight.
\end{enumerate}
\end{remark}

\section{Packing Radius and Covering Radius }

\quad\;In this section, we extend the concept of radius defined for the Hamming metric (see [\ref{HUFFMAN}]) to the case of weighted poset block metric. We always assume that $w$ is a weight on $\mathbb{F}_q$, $P=([s],\leq)$ is a poset, $\pi: [s]\rightarrow \mathbb{N}$ is a labeling of the poset $P$ and $V=\bigoplus\limits_{i=1}^s\mathbb{F}_q^{k_i}$ which is isomorphic to $\mathbb{F}_q^n$.

\begin{Definition}
Let $w$ be a weight on $\mathbb{F}_q$. For $u\in V$, the \emph{$(P,\pi,w)$-ball} with center $u$ and radius $r$ is the set
$$B_{w,(P,\pi)}(u,r)=\{v\in V: d_{w,(P,\pi)}(u,v)\leq r\}.$$
When the wight $w$ over $\mathbb{F}_q$ is considered to be Hamming weight, we denote by $B_{(P,\pi)}(u,r)$ the $(P,\pi,w)$-ball with center $u$ and radius $r$.
\end{Definition}

\begin{Definition}
Let $C$ be a $(P,\pi,w)$-code. The covering radius $\tilde{\rho}(C)$ is the smallest integer $l$ such that $V$ is the union of the balls with radius $l$ centered at the codewords of $C$, that is:
$$\tilde{\rho}(C)=\max\limits_{\boldsymbol{v}\in V}\min\limits_{\boldsymbol{u}\in C}d_{w,(P,\pi)}(\boldsymbol{u},\boldsymbol{v}).$$
\end{Definition}

\begin{Definition}
Let $C$ be a linear $\left(P,\pi,w\right)$-code and $\boldsymbol{v}\in V$. Define the coset of $C$ determined by $\boldsymbol{v}$ is defined as $\boldsymbol{v}+C=\{\boldsymbol{v}+\boldsymbol{u}: \boldsymbol{u}\in C\}$. The weight of a coset is the smallest weight of all vectors in the coset, and any vector having the smallest weight in the coset is called a coset leader.
\end{Definition}

\begin{remark}
For a linear $\left(P,\pi,w\right)$-code $C$,  one has that $\tilde{\rho}(C)$ is the weight of a coset with the largest weight.
\end{remark}

\begin{Definition}
A code $C$ is said to be an \emph{$r$-perfect $(P,\pi,w)$-code} if the $(P,\pi,w)$-balls of radius $r$ centered at the codewords of $C$ are pairwise disjoint and their union is $V$.
\end{Definition}

\begin{Definition}
The \emph{packing radius} $\rho(C)$ of a code $C$ is the largest radius of balls centered at codewords so that the balls are pairwise disjoint. We call a code $C$ is \emph{perfect} if it is $\rho(C)$-perfect.
\end{Definition}

In the remainder of this section, we always suppose that $P$ is a chain. Without loss of generality, we may assume that $P$ has order relation $1<\cdots <s$.

\begin{lemma}\label{ball}
Let $w$ be a weight on $\mathbb{F}_q$ and let $r=s+iM_w$ where $s\in[M_w]$ and $i\geq 0$ be an integer. Then $B_{w, (P,\pi)}(\boldsymbol{0},r)\subseteq B_{(P,\pi)}(\boldsymbol{0},i+1)$, with equality holds if and only if $s=M_w$.
\end{lemma}

\begin{proof}
Let $\boldsymbol{u}\in B_{w, (P,\pi)}(\boldsymbol{0},r)$. Then $\overline{\omega}_{w,(P,\pi)}(\boldsymbol{u})\leq r=s+iM_w\leq (i+1)M_w$. Therefore $\boldsymbol{u_j}=\boldsymbol{0}$ for $j\geq i+2$ and hence $\boldsymbol{u}\in B_{(P,\pi)}(\boldsymbol{0},i+1)$.

Suppose that $B_{w, (P,\pi)}(\boldsymbol{0},r)= B_{(P,\pi)}(\boldsymbol{0},i+1)$. Take $\boldsymbol{u}\in B_{(P,\pi)}(\boldsymbol{0},i+1)$ such that $u_{i+1,1}=\alpha$ where $\alpha\in\mathbb{F}_q$ satisfies $w(\alpha)=M_w$. Then $\overline{\omega}_{w,(P,\pi)}(\boldsymbol{u})=M_w+iM_w$. Therefore $s=M_w$.

Conversely suppose $s=M_w$. For $\boldsymbol{u}\in B_{(P,\pi)}(\boldsymbol{0},i+1)$, we have $\overline{\omega}_{w,(P,\pi)}(\boldsymbol{u})\leq (i+1)M_w=r$. Therefore $\boldsymbol{u}\in B_{w, (P,\pi)}(\boldsymbol{0},r)$ and hence $ B_{(P,\pi)}(\boldsymbol{0},i+1)\subseteq B_{w, (P,\pi)}(\boldsymbol{0},r)$.
\end{proof}

When $w$ is taken to be the Hamming weight over $\mathbb{F}_q$, the $(P,\pi,w)$-weight reduces to the NRT block weight. It is known that the packing radius of a linear $(n,K)$ code $C$ under NRT block metric is
$$\rho(C)=d_{(P,\pi)}(C)-1$$
(see [\ref{panek}], Theorem 5).

By Lemma \ref{ball}, we immediately get the following result.
\begin{theorem}
The packing radius of a linear $(P,\pi,w)$-code $C$ satisfies that
$$\rho(C)\geq\left(d_{(P,\pi)}(C)-1\right)M_w.$$
Furthermore, $\rho(C)=\left(d_{(P,\pi)}(C)-1\right)M_w$ if and only if $d_{w,(P,\pi)}(C)=m_w+\left(d_{(P,\pi)}(C)-1\right)M_w$.
\end{theorem}

Let $C$ be a  $(P,\pi,w)$-code, denote by $C_i=\{\boldsymbol{u_i}:\boldsymbol{u}\in C\}$ for $i\in[s]$.

\begin{theorem}
Let $C$ be a linear $(P,\pi,w)$-code. Set
$$r=\left\{
                             \begin{array}{ll}
                             s &\text{if}\ C_s\neq \mathbb{F}_q^{k_s};\\[2mm]
                             \min\left\{l: (C_{l+1},\ldots,C_s)= \mathbb{F}_q^{k_{l+1}}\oplus\cdots\oplus\mathbb{F}_q^{k_s}\right\} &\text{otherwise}.
                             \end{array}
                           \right.$$
Then
$$(r-1)M_w<\tilde{\rho}(C)\leq rM_w.$$
\end{theorem}

\begin{proof}
Suppose that $C_s\neq \mathbb{F}_q^{k_s}$. Then there exists $\boldsymbol{v}\in\mathbb{F}_q^n$ such that $\boldsymbol{v_s}\in\mathbb{F}_q^{k_s}\setminus C_s$. Then $(s-1)M_w+m_w\leq d_{w,(P,\pi)}(\boldsymbol{v},\boldsymbol{u})$ for any $\boldsymbol{u}\in C$ which implies that $(s-1)M_w<\tilde{\rho}(C)\leq sM_w$.

Suppose that $r=\min\left\{l: (C_{l+1},\ldots,C_s)= \mathbb{F}_q^{k_{l+1}}\oplus\cdots\oplus\mathbb{F}_q^{k_s}\right\}$. Take $\boldsymbol{v}\in\mathbb{F}_q^n$ such that $\boldsymbol{v_r}\in\mathbb{F}_q^{k_r}\setminus C_r$, then $d_{w,(P,\pi)}(\boldsymbol{v},\boldsymbol{u})\geq (r-1)M_w+m_w$ which implies that $\tilde{\rho}(C)\geq (r-1)M_w+m_w.$ On the other hand, for any $\boldsymbol{v}=(\boldsymbol{v_1},\ldots,\boldsymbol{v_r},\boldsymbol{v_{r+1}}, \ldots,\boldsymbol{v_s})\in\mathbb{F}_q^n$, there exists $\boldsymbol{u}\in C$ such that $\boldsymbol{u}=(\boldsymbol{u_1},\ldots,\boldsymbol{u_r},\boldsymbol{v_{r+1}}, \ldots,\boldsymbol{v_s})$. Therefore
$$d_{w,(P,\pi)}(\boldsymbol{u},\boldsymbol{v})= \overline{\omega}_{w,(P,\pi)}(\boldsymbol{u_1}-\boldsymbol{v_1},\ldots, \boldsymbol{u_l}-\boldsymbol{v_l},\boldsymbol{0},\ldots,\boldsymbol{0})\leq rM_w.$$
\end{proof}

\section{Code Constructions}

\quad\; In this section, we give several different ways to construct new $(P,\pi,w)$-codes from given.

\subsection{Construction 1}

\quad\;Let $P,Q$ be two posets and let $w$ be a weight on $\mathbb{F}_q$. Let $\pi_1: [s]\rightarrow\mathbb{N}$ such that $n_1=\sum\limits_{i=1}^s\pi_1(i)$ be a labeling of $P$ and let $\pi_2: [t]\rightarrow\mathbb{N}$ such that $\sum\limits_{i=1}^{t}\pi_2(i)=n_2$ be a labeling of $Q$.  Let $C_1\subseteq \left(\mathbb{F}_q^{n_1},d_{w,(P,\pi_1)}\right)$ be a $(P,\pi,w)$-code and $C_2\subseteq \left(\mathbb{F}_q^{n_2},d_{w,(Q,\pi_2)}\right)$ be an $(Q,\pi_2,w)$-code.   The \emph{direct sum} of $C_1$ and $C_2$ is defined as
$$\mathcal {C}=C_1\oplus C_2=\left\{\left(\boldsymbol{u'},\boldsymbol{u''}\right):\boldsymbol{u'}\in C_1,\ \boldsymbol{u''}\in C_2\right\}.$$
Define the \emph{direct sum of labeling map} $\pi_1$ and $\pi_2$ as $\pi=\pi_1\oplus\pi_2:[s+t]\rightarrow \mathbb{N}$ such that
$$\pi(i)=\left\{
                             \begin{array}{ll}
                             \pi_1(i) &\text{if}\ i\leq s;\\[2mm]
                             \pi_2(i-s)&\text{if}\ i>s.
                             \end{array}
                           \right.$$
Suppose that $\mathcal {L}=P\uplus Q$ (or $\mathcal {L}=P\oplus Q$).
For $\boldsymbol{u}=\left(\boldsymbol{u'},\boldsymbol{u''}\right)\in \mathcal {C}$ where $\boldsymbol{u'}\in C_1$ and $\boldsymbol{u''}\in C_2$, define
$$W_i^{\mathcal {L}}(\boldsymbol{u})=\left\{
                             \begin{array}{ll}
                             W_i^P\left(\boldsymbol{u'}\right) &\text{if}\ i\leq s;\\[2mm]
                            W_{i-s}^Q\left(\boldsymbol{u''}\right) &\text{if}\ i>s.
                             \end{array}
                           \right.$$

With notations introduced above, we obtain the following result.
\begin{proposition}\label{DIRECT}
(1) The code $\mathcal {C}\subseteq\left(\mathbb{F}_q^{n_1+n_2},d_{w,(P\uplus Q,\pi_1\oplus\pi_2)}\right)$ is a $\left(P\uplus Q,\pi_1\oplus\pi_2,w\right)$-code such that
$$d_{w,(P\uplus Q,\pi_1\oplus\pi_2)}(\mathcal {C})=\min\left\{d_{w,(P,\pi_1)}(C_1),d_{w,(Q,\pi_2)}(C_2)\right\}.$$
(2) The code $\mathcal {C}\subseteq\left(\mathbb{F}_q^{n_1+n_2},d_{w,(P\oplus Q,\pi_1\oplus\pi_2)}\right)$ is a $\left(P\oplus Q,\pi_1\oplus\pi_2,w\right)$-code such that
$$d_{w,(P\oplus Q,\pi_1\oplus\pi_2)}(\mathcal {C})=d_{w,(P,\pi_1)}(C_1).$$
\end{proposition}

\begin{proof}
Take $\boldsymbol{u}=(\boldsymbol{u'},\boldsymbol{u''})$ and $\boldsymbol{v}=(\boldsymbol{v',v''})$ where $\boldsymbol{u'}, \boldsymbol{v'}\in C_1$ and $\boldsymbol{u''}, \boldsymbol{v''}\in C_2$ respectively. It follows from the definition of weighted poset block weight that
$$\overline{\omega}_{w,(\mathcal {L},\pi)}(\boldsymbol{u}-\boldsymbol{v})=\sum\limits_{i\in M_{\boldsymbol{u-v}}^{\mathcal {L}}}W_i^{\mathcal {L}}(\boldsymbol{u}-\boldsymbol{v})+\sum\limits_{i\in I_{\boldsymbol{u-v}}^{\mathcal {L}}\setminus M_{\boldsymbol{u-v}}^{\mathcal {L}}} M_w.$$

Consider the weighted poset block metric under the poset $\mathcal {L}=P\uplus Q$, we have $I_{\boldsymbol{u-v}}^{\mathcal {L}}=I_{\boldsymbol{u'}-\boldsymbol{v'}}^{P}\cup I_{\boldsymbol{u''-v''}}^{Q}$ and $ M_{\boldsymbol{u-v}}^{\mathcal {L}}=M_{\boldsymbol{u'}-\boldsymbol{v'}}^{P}\cup M_{\boldsymbol{u''-v''}}^{Q}$.
Hence
\begin{eqnarray*}
\overline{\omega}_{w,(\mathcal {L},\pi)}(\boldsymbol{u}-\boldsymbol{v})
&=&\sum\limits_{i\in I_{\boldsymbol{u'}-\boldsymbol{v'}}^{P}} W_i^P\left(\boldsymbol{u'}-\boldsymbol{v'}\right)+ \sum\limits_{i\in I_{\boldsymbol{u''}-\boldsymbol{v''}}^{P}} W_i^Q\left(\boldsymbol{u''}-\boldsymbol{v''}\right)\\
&&+\sum\limits_{i\in I_{\boldsymbol{u'-v'}}^{P}\setminus M_{\boldsymbol{u'-v'}}^P}M_w+\sum\limits_{i\in I_{\boldsymbol{u''-v''}}^{Q}\setminus M_{\boldsymbol{u''-v''}}^Q}M_w\\[2mm]
&=&\overline{\omega}_{w,(P,\pi_1)}\left(\boldsymbol{u'}-\boldsymbol{v'}\right)+ \overline{\omega}_{w,(Q,\pi_2)}\left(\boldsymbol{u''}-\boldsymbol{v''}\right).
\end{eqnarray*}

If $\mathcal {L}=P\oplus Q$, then
$$\begin{array}{ccc}
\hline
    & \boldsymbol{u''}\neq\boldsymbol{v''} & \boldsymbol{u''}=\boldsymbol{v''}  \\[2mm]
  \hline
  I_{\boldsymbol{u-v}}^{\mathcal {L}} & P\cup I_{\boldsymbol{u''-v''}}^{\mathcal {Q}} & I_{\boldsymbol{u'-v'}}^{\mathcal {P}} \\[2mm]
  M_{\boldsymbol{u-v}}^{\mathcal {L}} & M_{\boldsymbol{u''-v''}}^{Q} & M_{\boldsymbol{u'-v'}}^{P}.\\
  \hline
\end{array}$$
Hence
$$\overline{\omega}_{w,(\mathcal {L},\pi)}(\boldsymbol{u}-\boldsymbol{v})=\left\{
                             \begin{array}{ll}
                             sM_w+\overline{\omega}_{w,(Q,\pi_2)} \left(\boldsymbol{u''}-\boldsymbol{v''}\right) &\text{if}\ \boldsymbol{u''}\neq\boldsymbol{v''};\\[2mm]
                            \overline{\omega}_{w,(P,\pi_1)} \left(\boldsymbol{u'}-\boldsymbol{v'}\right)  &\text{if}\ \boldsymbol{u''}=\boldsymbol{v''}.
                             \end{array}
                           \right.$$
The result then follows.
\end{proof}

\begin{lemma}\label{coset}
Suppose that $C_1$ is a linear $(P,\pi_1,w)$-code and $C_2$ is a linear $(Q,\pi_2,w)$-code. Let $\boldsymbol{u'},\boldsymbol{u''}$ be a coset leader of $C_1,C_2$ respectively. Then
\begin{enumerate}[(1)]
\item
 $\boldsymbol{u}=\left(\boldsymbol{u'},\boldsymbol{u''}\right)\in\mathcal {C}$ is a coset leader of $\mathcal {C}$ when $\mathcal {L}=P\uplus Q$.
\item
 $\boldsymbol{u}=\left(\boldsymbol{u'},\boldsymbol{u''}\right)\in\mathcal {C}$ is a coset leader of $\mathcal {C}$ when $\mathcal {L}=P\oplus Q$.
\end{enumerate}
\end{lemma}

\begin{proof}
Let $\boldsymbol{v}=\left(\boldsymbol{u'}+\boldsymbol{x'}, \boldsymbol{u''}+\boldsymbol{x''}\right)
\in \boldsymbol{u}+\mathcal {C}$ where $\left(\boldsymbol{x'},\boldsymbol{x''}\right)\in\mathcal {C}$.

If $\mathcal {L}=P\uplus Q$, then
$$\overline{\omega}_{w,(\mathcal {L},\pi)}(\boldsymbol{v})= \overline{\omega}_{w,(P,\pi_1)}\left(\boldsymbol{u'}+\boldsymbol{x'}\right)+ \overline{\omega}_{w,(Q,\pi_2)}\left(\boldsymbol{u''}+\boldsymbol{x''}\right) \geq\overline{\omega}_{w,(P,\pi_1)}\left(\boldsymbol{u'}\right)+ \overline{\omega}_{w,(Q,\pi_2)}\left(\boldsymbol{u''}\right)= \overline{\omega}_{w,(\mathcal {L},\pi)}(\boldsymbol{u})$$
which implies that $\boldsymbol{u}$ is a coset leader of $\boldsymbol{u}+\mathcal {C}$.

If $\mathcal {L}=P\oplus Q$, then we have
$$\begin{array}{ccc}
\hline
    & \boldsymbol{u''}+\boldsymbol{x''}\neq\boldsymbol{0} & \boldsymbol{u''}+\boldsymbol{x''}=\boldsymbol{0}  \\[2mm]
  \hline
  I_{\boldsymbol{v}}^{\mathcal {L}} & P\cup I_{\boldsymbol{u''+x''}}^{Q} & I_{\boldsymbol{u'+x'}}^P \\[2mm]
  M_{\boldsymbol{v}}^{\mathcal {L}} & M_{\boldsymbol{u''+x''}}^{Q} & M_{\boldsymbol{u'+x'}}^{P}.\\
  \hline
\end{array}$$
Hence
$$\overline{\omega}_{w,(\mathcal {L},\pi)}(\boldsymbol{v})=\left\{
                             \begin{array}{ll}
                             sM_w+\overline{\omega}_{w,(Q,\pi_2)} \left(\boldsymbol{u''}+\boldsymbol{x''}\right)\geq s M_w+\overline{\omega}_{w,(Q,\pi_2)} \left(\boldsymbol{u''}\right)= \overline{\omega}_{w,(\mathcal {L},\pi)}(\boldsymbol{u}) &\text{if}\ \boldsymbol{u^{''}}+\boldsymbol{v^{''}}\neq\boldsymbol{0};
                             \\[3mm]
                            \overline{\omega}_{w,(P,\pi_1)} \left(\boldsymbol{u'}-\boldsymbol{v'}\right)\geq  \overline{\omega}_{w,(P,\pi_1)}\left(\boldsymbol{u'}\right) &\text{if}\ \boldsymbol{u''}+\boldsymbol{x''}=\boldsymbol{0}.
                             \end{array}
                           \right.$$
Note that $\boldsymbol{u''}+\boldsymbol{x''}=\boldsymbol{0}$ implies that $\boldsymbol{u''}=-\boldsymbol{x''}\in C_2$ and hence $\boldsymbol{u''}=\boldsymbol{0}$. Therefore
$$\overline{\omega}_{w,(\mathcal {L},\pi)}(\boldsymbol{v})\geq\overline{\omega}_{w,(P,\pi_1)} \left(\boldsymbol{u'}\right)=\overline{\omega}_{w,(\mathcal {L},\pi)}(\boldsymbol{u}).$$
The result then follows.
\end{proof}

\begin{theorem}
Let $C_1$ be a linear $(P,\pi_1,w)$-code and $C_2$ be a linear $(Q,\pi_2,w)$-code. Then
\begin{enumerate}[(1)]
\item $\tilde{\rho}(\mathcal {C})=\tilde{\rho}(C_1)+\tilde{\rho}(C_2)$ when $\mathcal {L}=P\uplus Q$.
\item $\tilde{\rho}(\mathcal {C})=sM_w+\tilde{\rho}(C_2)$ when $\mathcal {L}=P\oplus Q$.
\end{enumerate}
\end{theorem}

\begin{proof}
Let $\mathcal {L}=P\oplus Q$. We first show that $\tilde{\rho}(\mathcal {C})\geq s Mw+\tilde{\rho}(C_2)$. Let $\boldsymbol{u'}$ be a coset leader of $C_1$ such that $\overline{\omega}_{w,(P,\pi_1)}\left(\boldsymbol{u'}\right)=\tilde{\rho}(C_1)$ and let $\boldsymbol{u''}$ be a coset leader of $C_2$ such that $\overline{\omega}_{w,(Q,\pi_2)}\left(\boldsymbol{u''}\right)=\tilde{\rho}(C_2)$. It follows from Lemma \ref{coset} that $\boldsymbol{u}=\left(\boldsymbol{u'},\boldsymbol{u''}\right)$ is a coset leader of $\mathcal {C}$. Then for any $\boldsymbol{c}\in\mathcal {C}$, we have
$$d_{w,(\mathcal {L},\pi)}(\boldsymbol{u},\boldsymbol{c})=\overline{\omega}_{w,(\mathcal {L},\pi)}(\boldsymbol{u}-\boldsymbol{c})\geq\overline{\omega}_{w,(\mathcal {L},\pi)}(\boldsymbol{u})=sM_w+\tilde{\rho}(C_2)$$
which implies that $\tilde{\rho}(\mathcal {C})\geq sM_w+\tilde{\rho}(C_2)$.

Conversely, let $\boldsymbol{z}=\left(\boldsymbol{z'},\boldsymbol{z''}\right)\in \mathbb{F}_q^{n_1+n_2}$. Then $\boldsymbol{z'}=\boldsymbol{u'}+\boldsymbol{x'}\in\boldsymbol{u'}+C_1$ where $\boldsymbol{u'}$ is a coset leader of $C_1$ and $\boldsymbol{z''}=\boldsymbol{u''}+\boldsymbol{x''}\in\boldsymbol{u''}+C_2$ where $\boldsymbol{u''}$ is a coset leader of $C_2$. Denote by $\boldsymbol{u}=\left(\boldsymbol{u'},\boldsymbol{u''}\right)$. Take $\boldsymbol{x}=\left(\boldsymbol{x'},\boldsymbol{x''}\right)\in\mathcal {C}$, we have
$$d_{w,(\mathcal {L},\pi)}(\boldsymbol{z},\boldsymbol{x})=\overline{\omega}_{w,(\mathcal {L},\pi)}(\boldsymbol{z}-\boldsymbol{x})=\overline{\omega}_{w,(\mathcal {L},\pi)}(\boldsymbol{z}-\boldsymbol{x})=\overline{\omega}_{w,(\mathcal {L},\pi)}(\boldsymbol{u})=s M_w+\overline{\omega}_{w,(Q,\pi_2)}\left(\boldsymbol{u''}\right)\leq s M_w+\tilde{\rho}(C_2).$$
Therefore $\min\limits_{\boldsymbol{u}\in\mathcal {C}}d_{w,(\mathcal {L},\pi)}(\boldsymbol{u},\boldsymbol{z})\leq s M_w+\tilde{\rho}(C_2)$ for any $\boldsymbol{z}\in\mathbb{F}_q^{n_1+n_2}$ and hence $\tilde{\rho}(\mathcal {C})\leq sM_w+\tilde{\rho}(C_2)$.

The case for $\mathcal {L}=P\uplus Q$ can be proved in the same way.
\end{proof}

\subsection{Construction 2}

\quad\;Let $C_1\subseteq\left(\mathbb{F}_q^n,d_{w,(P,\pi_1)}\right)$ be a $(P,\pi_1,w)$-code and let $C_2\subseteq\left(\mathbb{F}_q^n,d_{w,(Q,\pi_2)}\right)$ be a $(Q,\pi_2,w)$-code where $n=\sum\limits_{i=1}^{s}\pi_1(i)=\sum\limits_{i=1}^t\pi_2(i)$. Let $\pi$ be the direct sum of labeling $\pi_1$ and $\pi_2$. Let $\mathcal {L}=P\uplus Q$ (or $\mathcal {L}=P\oplus Q$). The \emph{$(\boldsymbol{u'}\ |\ \boldsymbol{u'}+\boldsymbol{u''})$ construction} produces the $(\mathcal {L},\pi,w)$-code
$$\mathcal {C}=\left\{\left(\boldsymbol{u'},\boldsymbol{u'}+\boldsymbol{u''}\right):\boldsymbol{u'}\in C_1, \boldsymbol{u''}\in C_2\right\}.$$

With the notations introduced above, we have the following result.

\begin{proposition}
\begin{enumerate}[(1)]
\item
The code $\mathcal {C}\subseteq \left(\mathbb{F}_q^{2n},d_{w,(P\uplus Q,\pi)}\right)$ is a $(P\uplus Q, \pi,w)$-code such that
$$d_{w,(P\uplus Q,\pi)}(\mathcal {C})\geq\min\left\{d_{w,(P,\pi_1)}(C_1),d_{w,(Q,\pi_2)}(C_2)\right\}$$
or
$$d_{w,(P\uplus Q,\pi)}(\mathcal {C})\geq\min\left\{d_{w,(Q,\pi_2)}(C_2),d_{w,(P,\pi_1)}(C_1)+d_{w,(Q,\pi_2)}(C_1), d_{w,(P,\pi_1)}(C_1)+d_{w,(Q,\pi_2)}(C_1+C_2)\right\}.$$
\item The code $\mathcal {C}\subseteq \left(\mathbb{F}_q^{2n},d_{w,(P\oplus Q,\pi)}\right)$ is a $(P\oplus Q, \pi,w)$-code such that
    $$d_{w,(P\oplus Q,\pi)}(\mathcal {C})\geq d_{(w,(P,\pi_1)}(C_1)$$
     or
    $$d_{w,(P\oplus Q,\pi)}(\mathcal {C})=\min\left\{d_{w,(Q,\pi_2)}(C_2), d_{w,(Q,\pi_2)}(C_1),d_{w,(Q,\pi_2)}(C_1+C_2)\right\}+sM_w$$
(here $C_1+C_2=\left\{\boldsymbol{u'}+\boldsymbol{u''}: \boldsymbol{u'}\in C_1,\boldsymbol{u''}\in C_2\right\}$).
\end{enumerate}
\end{proposition}

\begin{proof}
Let $\boldsymbol{u}=\left(\boldsymbol{u'},\boldsymbol{u'}+\boldsymbol{u''}\right)$ and $\boldsymbol{v}=\left(\boldsymbol{v'},\boldsymbol{v'}+\boldsymbol{v''}\right)$ where $\boldsymbol{u'},\boldsymbol{v'}\in C_1$ and $\boldsymbol{u''},\boldsymbol{v''}\in C_2$. Then
$$d_{w,(\mathcal {L},\pi)}(\boldsymbol{u},\boldsymbol{v})=\overline{\omega}_{w,(\mathcal {L},\pi)}(\boldsymbol{u}-\boldsymbol{v})=\sum\limits_{i\in M_{\boldsymbol{u-v}}^{\mathcal {L}}}W_i(\boldsymbol{u}-\boldsymbol{v})+M_w\left|I_{\boldsymbol{u-v}}^{\mathcal {L}}\setminus M_{\boldsymbol{u-v}}^{\mathcal {L}}\right|.$$

If $\mathcal {L}=P\uplus Q$, then
$$\begin{array}{ccccc}
\hline
    & \boldsymbol{u'}=\boldsymbol{v'} & \boldsymbol{u''}=\boldsymbol{v''}& \boldsymbol{u'}+\boldsymbol{u''}=\boldsymbol{v'}+\boldsymbol{v''}& \boldsymbol{u'}+\boldsymbol{u''}\neq\boldsymbol{v'}+\boldsymbol{v''} \\[2mm]
  \hline
  I_{\boldsymbol{u-v}}^{\mathcal {L}} & I_{\boldsymbol{u''-u''}}^{Q} & I_{\boldsymbol{u'-v'}}^{P}\cup I_{\boldsymbol{u'-v'}}^{Q}& I_{\boldsymbol{u'-v'}}^{P}& I_{\boldsymbol{u'-v'}}^{P}\cup I_{\boldsymbol{u'+u''-v'-v''}}^Q \\[2mm]
  M_{\boldsymbol{u-v}}^{\mathcal {L}} & M_{\boldsymbol{u''-v''}}^{Q} & M_{\boldsymbol{u'-v'}}^{P}\cup M_{\boldsymbol{u'-v'}}^{Q} & M_{\boldsymbol{u'-v'}}^{P}&M_{\boldsymbol{u'-v'}}^{P}\cup M_{\boldsymbol{u'+u''-v'-v''}}^Q. \\
  \hline
\end{array}$$
Therefore
$$\overline{\omega}_{w,(\mathcal {L},\pi)}(\boldsymbol{u}-\boldsymbol{v})=\left\{
                             \begin{array}{ll}
  \overline{\omega}_{w,(Q,\pi_2)}\left(\boldsymbol{u''}-\boldsymbol{v''}\right) & \text{if}\ \boldsymbol{u'}=\boldsymbol{v'};\\[2mm]
  \overline{\omega}_{w,(P,\pi_1)}\left(\boldsymbol{u'}-\boldsymbol{v'}\right)+ \overline{\omega}_{w,(Q,\pi_2)}\left(\boldsymbol{u'}-\boldsymbol{v'}\right)& \text{if}\ \boldsymbol{u''}=\boldsymbol{v''};\\[2mm]
  \overline{\omega}_{w,(P,\pi_1)}\left(\boldsymbol{u'}-\boldsymbol{v'}\right)& \text{if}\ \boldsymbol{u'}+\boldsymbol{u''}=\boldsymbol{v'}+\boldsymbol{v''};\\[2mm]
  \overline{\omega}_{w,(P,\pi_1)}\left(\boldsymbol{u'}-\boldsymbol{v'}\right)+ \overline{\omega}_{w,(Q,\pi_2)}\left(\boldsymbol{u'}+\boldsymbol{u''}- \boldsymbol{v'}-\boldsymbol{v''}\right)
  & \text{if}\ \boldsymbol{u'}+\boldsymbol{u''}\neq\boldsymbol{v'}+\boldsymbol{v''}.
                             \end{array}\right.$$
Hence
$$d_{w,(\mathcal {L},\pi)}(\mathcal {C})\geq\min\left\{d_{w,(Q,\pi_2)}(C_2),d_{w,(P,\pi_1)}(C_1)+d_{w,(Q,\pi_2)}(C_1), d_{w,(P,\pi_1)}(C_1)+d_{w,(Q,\pi_2)}(C_1+C_2)\right\}$$
if there does not exist $\boldsymbol{u'},\boldsymbol{v'}\in C_1$ and $\boldsymbol{u''},\boldsymbol{v''}\in C_2$  such that $\boldsymbol{u'}+\boldsymbol{u''}=\boldsymbol{v'}+\boldsymbol{v''}$. Otherwise we have that
$$d_{w,(\mathcal {L},\pi)}(\mathcal {C})\geq\min\left\{d_{w,(P,\pi_1)}(C_1),d_{w,(Q,\pi_2)}(C_2)\right\}.$$

If $\mathcal {L}=P\oplus Q$, then
$$\begin{array}{ccccc}
\hline
    & \boldsymbol{u'}=\boldsymbol{v'} & \boldsymbol{u''}=\boldsymbol{v''}& \boldsymbol{u'}+\boldsymbol{u''}=\boldsymbol{v'}+\boldsymbol{v''}& \boldsymbol{u'}+\boldsymbol{u''}\neq\boldsymbol{v'}+\boldsymbol{v''} \\[2mm]
  \hline
  I_{\boldsymbol{u-v}}^{\mathcal {L}} & P\cup I_{\boldsymbol{u''-u''}}^{Q} & P \cup I_{\boldsymbol{u'-v'}}^{Q}& I_{\boldsymbol{u'-v'}}^{P}& P\cup I_{\boldsymbol{u'+u''-v'-v''}}^Q \\[2mm]
  M_{\boldsymbol{u-v}}^{\mathcal {L}} & M_{\boldsymbol{u''-v''}}^{Q} & M_{\boldsymbol{u'-v'}}^{Q} & M_{\boldsymbol{u'-v'}}^{P}& M_{\boldsymbol{u'+u''-v'-v''}}^Q. \\
  \hline
\end{array}$$
Therefore
$$\overline{\omega}_{w,(\mathcal {L},\pi)}(\boldsymbol{u}-\boldsymbol{v})=\left\{
                             \begin{array}{ll}
  \overline{\omega}_{w,(Q,\pi_2)}\left(\boldsymbol{u''}-\boldsymbol{v''}\right)+s M_w & \text{if}\ \boldsymbol{u'}=\boldsymbol{v'};\\[2mm] \overline{\omega}_{w,(Q,\pi_2)}\left(\boldsymbol{u'}-\boldsymbol{v'}\right)+s M_w& \text{if}\ \boldsymbol{u''}=\boldsymbol{v''};\\[2mm]
  \overline{\omega}_{w,(P,\pi_1)}\left(\boldsymbol{u'}-\boldsymbol{v'}\right)& \text{if}\ \boldsymbol{u'}+\boldsymbol{u''}=\boldsymbol{v'}+\boldsymbol{v''};\\[2mm] \overline{\omega}_{w,(Q,\pi_2)}\left(\boldsymbol{u'}+\boldsymbol{u''}- \boldsymbol{v'}-\boldsymbol{v''}\right)+sM_w
  & \text{if}\ \boldsymbol{u'}+\boldsymbol{u''}\neq\boldsymbol{v'}+\boldsymbol{v''}.
                             \end{array}\right.$$
Hence
$$d_{w,(\mathcal {L},\pi)}(\mathcal {C})=\min\left\{d_{w,(Q,\pi_2)}(C_2), d_{w,(Q,\pi_2)}(C_1),d_{w,(Q,\pi_2)}(C_1+C_2)\right\}+s M_w$$
if there exists no $\boldsymbol{u'},\boldsymbol{v'}\in C_1$ and $\boldsymbol{u''},\boldsymbol{v''}\in C_2$ such that $\boldsymbol{u'}+\boldsymbol{u''}=\boldsymbol{v'}+\boldsymbol{v''}$.
\end{proof}

\begin{theorem}
Let $C_1$ be a linear $(P,\pi_1,w)$-code and $C_2$ be a linear $(Q,\pi_2,w)$-code. Then
\begin{enumerate}[(1)]
\item The code $\mathcal {C}$ is a linear $\left(P\uplus Q,\pi,w\right)$-code satisfies $\tilde{\rho}(\mathcal {C})\leq\tilde{\rho}(C_1)+\tilde{\rho}(C_2)$.
\item The code $\mathcal {C}$ is a linear $\left(P\oplus Q,\pi,w\right)$-code satisfies $\tilde{\rho}(\mathcal {C})\leq\tilde{\rho}(C_2)+sM_w$.
\end{enumerate}
\end{theorem}

\begin{proof}
Let $\boldsymbol{v}=\left(\boldsymbol{v'},\boldsymbol{v''}\right)\in \mathbb{F}_q^{2n}$ where $\boldsymbol{v'}, \boldsymbol{v''}\in\mathbb{F}_q^n$. Then $\boldsymbol{v'}=\boldsymbol{\alpha'}+\boldsymbol{a'}$ and $\boldsymbol{v''}=\boldsymbol{\beta'}+\boldsymbol{b'}$ where $\boldsymbol{a'}$, $\boldsymbol{b'}\in C_1$ and  $\boldsymbol{\alpha'}$, $\boldsymbol{\beta'}$ are two coset leaders of in the corresponding coset of $C_1$. Then
$$\boldsymbol{v}=\left(\boldsymbol{v'},\boldsymbol{v''}\right)= \left(\boldsymbol{\alpha'}+\boldsymbol{a'},\boldsymbol{\beta'}+\boldsymbol{b'}\right)=
\left(\boldsymbol{\alpha'}+\boldsymbol{a'}, \boldsymbol{\beta'}+\boldsymbol{a'}+\left(\boldsymbol{b'}-\boldsymbol{a'}\right)\right).$$

Assume that $\boldsymbol{\beta'}+\left(\boldsymbol{b'}-\boldsymbol{a'}\right)= \boldsymbol{\gamma''}+\boldsymbol{c''}$ where $\boldsymbol{c''}\in C_2$ and $\boldsymbol{\gamma''}$ is a coset leader in the corresponding coset of $C_2$. Then
$$\boldsymbol{v}=\left(\boldsymbol{\alpha'}+\boldsymbol{a'}, \boldsymbol{\gamma''}+\boldsymbol{a'}+ \boldsymbol{c''}\right)=\left(\boldsymbol{\alpha'},\boldsymbol{\gamma''}\right)+ \left(\boldsymbol{a'},\boldsymbol{a'}+\boldsymbol{c''}\right).$$
Denote by $\boldsymbol{c}=\left(\boldsymbol{a'},\boldsymbol{a'}+\boldsymbol{c''}\right)\in\mathcal {C}$ and
$\boldsymbol{u}=\left(\boldsymbol{\alpha'},\boldsymbol{\gamma''}\right)$, we have
$$d_{w,(\mathcal {L},\pi)}(\boldsymbol{v},\boldsymbol{c})=\overline{\omega}_{w,(\mathcal {L},\pi)}(\boldsymbol{v}-\boldsymbol{c})=\overline{\omega}_{w,(\mathcal {L},\pi)}(\boldsymbol{u}).$$
If $\mathcal {L}=P\uplus Q$, then
$$\overline{\omega}_{w,(\mathcal {L},\pi)}(\boldsymbol{u})=\overline{\omega}_{w,(P,\pi_1)}\left(\boldsymbol{\alpha'}\right) + \overline{\omega}_{w,(Q,\pi_2)}\left(\boldsymbol{\gamma''}\right)\leq \tilde{\rho}(C_1)+\tilde{\rho}(C_2).$$
If $\mathcal {L}=P\oplus Q$, then
$$\overline{\omega}_{w,(\mathcal {L},\pi)}(\boldsymbol{u})=\overline{\omega}_{w,(P,\pi_1)}\left(\boldsymbol{\alpha'}\right) + \overline{\omega}_{w,(Q,\pi_2)}\left(\boldsymbol{\gamma''}\right)\leq s M_w+\tilde{\rho}(C_2).$$
The result then follows.
\end{proof}

\subsection{Construction 3}

\quad\; Let $\mathcal {L}$ be a poset with underling set $[s]$ and $\pi$ be a labeling map of the poset $\mathcal {L}$ such that $\sum\limits_{i=1}^{s}\pi(i)=n$. Let $\mathcal {C}\subseteq \left(\mathbb{F}_q^n,d_{w,(\mathcal {L},\pi)}\right)$ be a $(\mathcal {L},\pi,w)$-code over $\mathbb{F}_q$.

The extended code $\widehat{\mathcal {C}}$ of $\mathcal {C}$ is defined as:
$$\widehat{\mathcal {C}}=\left\{\left(\boldsymbol{u},u_{s+1}\right): \boldsymbol{u}\in\mathcal {C},\ u_{s+1}\in\mathbb{F}_q\ \text{with}\ u_{11}+\cdots+u_{1\pi(1)}+\cdots+u_{s1}+\cdots+u_{s\pi(s)}+u_{s+1}=0\right\}.$$

Consider the extending poset $\mathcal {L}^{+}$ of $\mathcal {L}$ by adding an element $s+1$ in $\mathcal {L}$ and the labeling map $\pi^{+}:[t]\rightarrow \mathbb{N}$ of $\mathcal {L}^{+}$ such that $\pi^{+}(i)=\pi(i)$ for $i\leq s$ and $\pi^{+}(s+1)=1$. Define
$$W_{i}^{\mathcal {L}^{+}}\left(\boldsymbol{u^{+}}\right)=\left\{
                             \begin{array}{ll}
                             W_i^{\mathcal {L}}(\boldsymbol{u}) &\text{if}\ i\leq s;\\[2mm]
                             w(u_{s+1})&\text{if}\ i=s+1.
                             \end{array}
                           \right.$$

The following results can be proved in a routine way.

\begin{remark}
The extended code $\widehat{\mathcal {C}}\subseteq\left(\mathbb{F}_q^{n+1},d_{w,(\mathcal {L}^{+},\pi^{+})}\right) $ is a $(\mathcal {L}^{+},\pi^{+},w)$-code satisfying that
$$d_{w,(\mathcal {L},\pi)}(\mathcal {C})\leq d_{w,(\mathcal {L}^{+},\pi^{+})}(\widehat{\mathcal {C}})\leq d_{w,(\mathcal {L},\pi)}(\mathcal {C})+M_w.$$
\end{remark}

\begin{theorem}
Let $\mathcal {C}$ be a linear $(\mathcal {L},\pi,w)$-code over $\mathbb{F}_q$. The covering radius of the $\left(\mathcal {L}^{+},\pi^{+},w\right)$-code $\widehat{\mathcal {C}}$ satisfies $\tilde{\rho}(\mathcal {C})\leq\tilde{\rho}(\widehat{\mathcal {C}})\leq\tilde{\rho}(\mathcal {C})+M_w$.
\end{theorem}

\subsection{Construction 4}

\quad\;
Let $\mathcal {L}$ be a poset with underlying set $[s]$ and let $\pi$ be a labeling map of the poset $\mathcal {L}$ such that $\sum\limits_{i=1}^{s}\pi(i)=n$. Let $T\subseteq [s]$ be any set of $t$ blocks. Let $\mathcal {C}$ be an $\left[n,K,d_{w,(\mathcal {L},\pi)}(\mathcal {C})\right]$ code over $\mathbb{F}_q$. Puncturing $\mathcal {C}$ on $T$ gives a code over $\mathbb{F}_q$ of length $n-\sum\limits_{i\in T}\pi(i)$, called the punctured code of $\mathcal {C}$ and denoted by $\mathcal {C}_T$.

In this section, we fix $T=\{i\}$ for some $i\in[s]$. For $\boldsymbol{u}=(\boldsymbol{u_1},\ldots,\boldsymbol{u_s})\in\mathbb{F}_q^n$, define
$$\boldsymbol{u^{*}}
=(\boldsymbol{u_1},\ldots,\boldsymbol{u_{i-1}}, \boldsymbol{u_{i+1}},\ldots,\boldsymbol{u_s}).$$
The punctured code $\mathcal {C^{*}}$ is given by
$$\mathcal {C^{*}}=\left\{\boldsymbol{u^{*}}: \boldsymbol{u}\in\mathcal {C}\right\}.$$

Considering the puncturing poset $\mathcal {L}^{-}$ of $\mathcal {L}$ by deleting $i$ from $[s]$ and the labeling map $\pi^{-}:[s]\setminus\{i\}\rightarrow \mathbb{N}$ of $\mathcal {L}^{-}$ such that $\pi^{-}(j)=\pi(j)$ for $j\in[s]\setminus\{i\}$, we get the following result.

\begin{proposition}
The punctured code $\mathcal {C}^{*}\subseteq\left(\mathbb{F}_q^{n-1},d_{w,(\mathcal {L}^{-},\pi^{-})}\right)$ is a $(\mathcal {L}^{-},\pi^{-},w)$-code such that $$d_{w,(\mathcal {L}^{-},\pi^{-})}(\mathcal {C}^{*})\leq d_{w,(\mathcal {L},\pi)}(\mathcal {C}).$$
\end{proposition}

\begin{proof}
Let $\boldsymbol{u^{*}},\boldsymbol{v^{*}}\in \mathcal {C}^{*}$ whose corresponding vectors are $\boldsymbol{u},\boldsymbol{v}\in\mathcal {C}$ respectively. It follows from the definition of puncturing poset that
$$\begin{array}{cccc}
\hline
    & i\in M_{\boldsymbol{u-v}}^{\mathcal {L}} & i\in I_{\boldsymbol{u-v}}^{\mathcal {L}}\setminus M_{\boldsymbol{u-v}}^{\mathcal {L}}& i\notin I_{\boldsymbol{u-v}}^{\mathcal {L}} \\[2mm]
  \hline
  I_{\boldsymbol{u^{*}-v^{*}}}^{\mathcal {L}^{-}} & I_{\boldsymbol{u-v}}^{\mathcal {L}}\setminus\{i\} & I_{\boldsymbol{u-v}}^{\mathcal {L}}\setminus\{i\} & I_{\boldsymbol{u-v}}^{\mathcal {L}} \\[2mm]
  M_{\boldsymbol{u^{*}-v^{*}}}^{\mathcal {L}^{-}} & M_{\boldsymbol{u-v}}^{\mathcal {L}}\setminus\{i\} & M_{\boldsymbol{u-v}}^{\mathcal {L}} & M_{\boldsymbol{u-v}}^{\mathcal {L}}.\\
  \hline
\end{array}$$
Thus
\begin{eqnarray*}
d_{w,(P^{-},\pi^{-})}\left(\boldsymbol{u^{*}},\boldsymbol{v^{*}}\right)&=& \overline{\omega}_{w,(\mathcal {L}^{-},\pi^{-})}\left(\boldsymbol{u^{*}}-\boldsymbol{v^{*}}\right)=\sum\limits_{j\in M_{\boldsymbol{u^{*}-v^{*}}}^{\mathcal {L}^{-}}}W_i\left(\boldsymbol{u^{*}}-\boldsymbol{v^{*}}\right)+ \left|I_{\boldsymbol{u^{*}-v^{*}}}^{\mathcal {L}^{-}}\setminus M_{\boldsymbol{u^{*}-v^{*}}}^{\mathcal {L}^{-}}\right|M_w\\
&=&\left\{
\begin{array}{ll}
\overline{\omega}_{w,(\mathcal {L},\pi)}(\boldsymbol{u}-\boldsymbol{v})-W_i(\boldsymbol{u}-\boldsymbol{v})&\text{if}\ i\in M_{\boldsymbol{u-v}}^{\mathcal {L}};\\[2mm]
\overline{\omega}_{w,(\mathcal {L},\pi)}(\boldsymbol{u}-\boldsymbol{v})-M_w &\text{if}\ i\in I_{\boldsymbol{u-v}}^{\mathcal {L}}\setminus M_{\boldsymbol{u-v}}^{\mathcal {L}};\\[2mm]
\overline{\omega}_{w,(\mathcal {L},\pi)}(\boldsymbol{u}-\boldsymbol{v})&\text{if}\ i\notin I_{\boldsymbol{u-v}}^{\mathcal {L}}.
\end{array}
\right.
\end{eqnarray*}
Hence $d_{w,(P^{-},\pi^{-})}(\mathcal {C}^{*})\leq d_{w,(\mathcal {L},\pi)}(\mathcal {C})$.
\end{proof}

\begin{remark}\label{weight}
From the above proof, we have that
$$\overline{\omega}_{w,(\mathcal {L}^{-},\pi^{-})}(\boldsymbol{u^{*}})\leq\overline{\omega}_{w,(\mathcal {L},\pi)}(\boldsymbol{u})$$
for any $\boldsymbol{u}\in\mathbb{F}_q^n$ such that $\boldsymbol{u^{*}}$ is the punctured vector of $\boldsymbol{u}$ on $i$-th block.
\end{remark}

\begin{theorem}
Let $\mathcal {C}$ be a linear $(\mathcal {L},\pi,w)$-code over $\mathbb{F}_q$. The punctured code $\mathcal {C}^{*}$ of $\mathcal {C}$ satisfies $\tilde{\rho}(\mathcal {C}^{*})\leq\tilde{\rho}(\mathcal {C})$.
\end{theorem}

\begin{proof}
Let $\boldsymbol{v}\in\mathbb{F}_q^n$. Then there exist $\boldsymbol{u}\in\mathcal {C}$ and $\boldsymbol{\alpha}\in\mathbb{F}_q^n$ a coset leader of $\mathcal {C}$ such that $\boldsymbol{v}=\boldsymbol{\alpha}+\boldsymbol{u}$. It follows from Remark \ref{weight} that
$$d_{w,(\mathcal {L}^{-},\pi^{-})}\left(\boldsymbol{v^{*},\boldsymbol{u^{*}}}\right)= \overline{\omega}_{w,(\mathcal {L}^{-},\pi^{-})}\left(\boldsymbol{v^{*}-\boldsymbol{u^{*}}}\right)\leq \overline{\omega}_{w,(\mathcal {L},\pi)}(\boldsymbol{v}-\boldsymbol{u})=\overline{\omega}_{w,(\mathcal {L},\pi)}(\boldsymbol{\alpha})\leq \tilde{\rho}(\mathcal {C}).$$
Since $\boldsymbol{v}\in\mathbb{F}_q^n$ is arbitrary, we have $\tilde{\rho}(\mathcal {C}^{*})\leq\tilde{\rho}(\mathcal {C})$.
\end{proof}

\subsection{Construction 5}

\quad\;
Let $P$ and $Q$ be two posets with underlining sets $[s]$ and $[t]$ respectively. Let $\pi_1:[s]\rightarrow\mathbb{N}$ be a labeling map of $P$ such that $\sum\limits_{i\in[s]}\pi_1(i)=n_1$ and let $\pi_2:[t]\rightarrow\mathbb{N}$ be a labeling map of $Q$ such that $\sum\limits_{i\in[t]}\pi_2(i)=n_2$. Suppose that $\mathcal {L}=P\otimes Q$ (or $\mathcal {L}=P\star Q$). Then $\mathcal {L}$ is a poset with underlying set $[s]\times[t]=\left\{(i,j):i\in[s],j\in[j]\right\}$ and cardinality $st$.
Denote by $\pi_1(i)=\alpha_i$ and $\pi_2(i)=\beta_i$ in the remainder of this section.

Define the \emph{direct product} of labeling map $\pi_1$ and $\pi_2$ as $\pi=\pi_1\otimes\pi_2:[s]\times[t]\rightarrow \mathbb{N}$ such that $$\pi((i,j))=\alpha_i\beta_j$$
for $(i,j)\in[s]\times [t]$. Then $\pi$ is a labeling map of $\mathcal {L}$ such that $\sum\limits_{(i,j)\in\mathcal {L}}\pi((i,j))=n_1n_2$.

Let $\boldsymbol{u}=(\boldsymbol{u_1},\ldots,\boldsymbol{u_s})\in  \left(\mathbb{F}_q^{n_1},d_{w,(P,\pi_1)}\right)$ where $\boldsymbol{u_i}\in\mathbb{F}_q^{\alpha_i}$ and $\boldsymbol{v}=(\boldsymbol{v_1},\ldots,\boldsymbol{v_t})\in  \left(\mathbb{F}_q^{n_2},d_{w,(Q,\pi_2)}\right)$ where $\boldsymbol{v_i}\in\mathbb{F}_q^{\beta_i}$. Define $\boldsymbol{u}\otimes\boldsymbol{v}$ as $$\left\{u_{ij}v_{rl}:i\in[s],j\in[\alpha_i],r\in[t],l\in[\beta_l]\right\} \in\mathbb{F}_q^{n_1n_2}.$$
We write it in the form of a block matrix as following:

\begin{center}
  $\boldsymbol{u}\otimes \boldsymbol{v}=\left[
     \begin{array}{cccc}
       G_{1,1}^{\boldsymbol{uv}} & G_{1,2}^{\boldsymbol{uv}} & \cdots & G_{1,t}^{\boldsymbol{uv}}\\[1mm]
       G_{2,1}^{\boldsymbol{uv}} & G_{2,2}^{\boldsymbol{uv}} & \cdots & G_{2,t}^{\boldsymbol{uv}}\\[1mm]
       \vdots &\vdots&&\vdots\\[1mm]
       G_{s,1}^{\boldsymbol{uv}} & G_{s,2}^{\boldsymbol{uv}} & \cdots & G_{s,t}^{\boldsymbol{uv}}
     \end{array}
   \right]$
\end{center}
where $G_{i,j}^{\boldsymbol{uv}}$ is an $\alpha_i\times\beta_j$ matrix for $i\in[s]$ and $j\in[t]$, that is:
\begin{center}
  $G_{i,j}^{\boldsymbol{uv}}=\left[
     \begin{array}{cccc}
     u_{i1}v_{j1}&u_{i1}v_{j2}&\cdots&u_{i1}v_{j\beta_j}\\[1mm]
     u_{i2}v_{j1}&u_{i2}v_{j2}&\cdots&u_{i2}v_{j\beta_j}\\[1mm]
     \vdots&\vdots&&\vdots\\[1mm]
     u_{i\alpha_i}v_{j1}&u_{i\alpha_i}v_{j2}&\cdots&u_{i\alpha_i}v_{j\beta_j}
     \end{array}
   \right].$
\end{center}
Now we has given a partition of $\boldsymbol{u}\otimes\boldsymbol{v}$ whose $(i,j)$-th block is $G^{\boldsymbol{uv}}_{i,j}$ corresponding to the element $(i,j)$ of the poset $\mathcal {L}$. Set
$$W_{ij}^{\boldsymbol{uv}}=\max\left\{w(u_{i\varsigma}v_{j\mu}): 1\leq \varsigma\leq\alpha_i,1\leq \mu\leq\beta_j\right\}.$$
For any $\boldsymbol{u}\in\mathbb{F}_q^{n_1n_2}$ with $st$ blocks, the $(\mathcal {L},\pi,w)$-weight of $\boldsymbol{u}$ is then
$$\overline{\omega}_{w,(\mathcal {L},\pi)}(\boldsymbol{u})=\sum\limits_{(i,j)\in M_{\boldsymbol{u}}^{\mathcal {L}}}W_{ij}^{\boldsymbol{u}\otimes\boldsymbol{1}}+\sum\limits_{(i,j)\in I_{\boldsymbol{u}}^{\mathcal {L}}\setminus M_{\boldsymbol{u}}^{\mathcal {L}}}M_w=\sum\limits_{(i,j)\in M_{\boldsymbol{u}}^{\mathcal {L}}}W_{ij}^{\boldsymbol{u}\otimes\boldsymbol{1}}+\left|(i,j)\in I_{\boldsymbol{u}}^{\mathcal {L}}\setminus M_{\boldsymbol{u}}^{\mathcal {L}}\right|\cdot M_w.$$

Let $C_1\subseteq\left(\mathbb{F}_q^{n_1},d_{w,(P,\pi_1)}\right)$ be a $(P,\pi_1,w)$-code and let $C_2\subseteq\left(\mathbb{F}_q^{n_2},d_{w,(Q,\pi_2)}\right)$ be a $(Q,\pi_2,w)$-code. The \emph{tensor product} of $C_1$ and $C_2$, denoted by $\mathcal {C}=C_1\bigotimes C_2$, is given by
$$C_1\bigotimes C_2=\left\{\boldsymbol{u}\otimes\boldsymbol{v}:\boldsymbol{u}\in C_1,\ \boldsymbol{v}\in C_2\right\}.$$

\begin{proposition}\label{CAR}
Let $C_1$ be a linear $(P,\pi_1,w)$-code and let $C_2$ be a linear $(Q,\pi_2,w)$-code. Let $\mathcal {L}=P\otimes Q$ and $\pi=\pi_1\otimes\pi_2$. Then the following results hold:
\begin{enumerate}[(1)]
\item Suppose that $P$ is a chain with order relation $1<2<\cdots<s$ and $Q$ is an antichain, then
    $$d_{(Q,\pi_2)}(C_2)\left(d_{(P,\pi_1)}(C_1)-1\right)M_w+d_{w,(Q,\pi_2)}(C_2)\leq d_{w,(\mathcal {L},\pi)}(\mathcal {C})\leq d_{(P,\pi_1)}(C_1)d_{(Q,\pi_2)}(C_2)M_w.$$
\item Suppose that $P$ and $Q$ are both antichains, then
$$d_{(P,\pi_1)}(C_1)d_{(Q,\pi_2)}(C_2) m_w\leq d_{w,(\mathcal {L},\pi)}(\mathcal {C})\leq d_{(P,\pi_1)}(C_1)d_{(Q,\pi_2)}(C_2)M_w.$$
\item Suppose that $P$ is a chain with order relation $1<2<\cdots<s$ and $Q$ is a chain with order relation $1<2<\cdots<t$, then
    $$\left(d_{(P,\pi_1)}(C_1)d_{(Q,\pi_2)}(C_2)-1\right)M_w+m_w\leq d_{w,(\mathcal {L},\pi)}(\mathcal {C})\leq d_{(P,\pi_1)}(C_1)d_{(Q,\pi_2)}(C_2)M_w.$$

\end{enumerate}

\end{proposition}

\begin{proof}
Let $\boldsymbol{u}\otimes\boldsymbol{v}\in\mathcal {C}$.
\begin{enumerate}[(1)]
\item If $P$ is a chain and $Q$ is an antichain, then $(i,j)\leq (i',j')\in\mathcal {L}$ if and only if $i\leq i'$ and $j=j'$. Assume that $d_{(P,\pi_1)}(\boldsymbol{u})=\lambda$ and $I_{\boldsymbol{v}}^Q=\{\eta_1,\eta_2,\ldots,\eta_r\}$. Then
\begin{center}
  $\boldsymbol{u}\otimes \boldsymbol{v}=\left[
     \begin{array}{ccccccccc}
       O & \cdots  & G_{1,\eta_1}^{\boldsymbol{uv}}&\cdots& G_{1,\eta_2}^{\boldsymbol{uv}}&\cdots& G_{1,\eta_r}^{\boldsymbol{uv}}&\cdots& O\\[1mm]
       O & \cdots  & G_{2,\eta_1}^{\boldsymbol{uv}}&\cdots& G_{2,\eta_2}^{\boldsymbol{uv}}&\cdots& G_{2,\eta_r}^{\boldsymbol{uv}}&\cdots& O\\[1mm]
       \vdots&&\vdots&&\vdots&&\vdots&&\vdots\\[1mm]
       O & \cdots  & G_{\lambda, \eta_1}^{\boldsymbol{uv}}&\cdots& G_{\lambda, \eta_2}^{\boldsymbol{uv}}&\cdots& G_{\lambda, \eta_r}^{\boldsymbol{uv}}&\cdots& O\\[1mm]
       O&\cdots&O&\cdots&O&\cdots&O&\cdots&O\\[1mm]
       \vdots&&\vdots&&\vdots&&\vdots&&\vdots\\[1mm]
       O&\cdots&O&\cdots&O&\cdots&O&\cdots&O\\[1mm]
     \end{array}
   \right]$
\end{center}
satisfies $G_{\lambda, \eta_l}^{\boldsymbol{uv}}\neq O$ for $1\leq l\leq r$ is an $\alpha_{\lambda}\times\beta_{\eta_l}$ matrix. Then
$$\overline{\omega}_{w,(\mathcal {L},\pi)}(\boldsymbol{u}\otimes\boldsymbol{v})= \sum\limits_{(i,j)=(\lambda,\eta_l),\atop 1\leq l\leq r}W_{ij}^{\boldsymbol{uv}}+(\lambda-1)j_r M_w=\sum\limits_{(i,j)=(\lambda,\eta_l),\atop 1\leq l\leq r}W_{ij}^{\boldsymbol{uv}}+\left(d_{(P,\pi_1)}(\boldsymbol{u})-1\right) d_{(Q,\pi_2)}(\boldsymbol{v})M_w.$$
Note that
$$\overline{\omega}_{w,(Q,\pi_2)}(\boldsymbol{v})=\sum\limits_{j=\eta_l,1\leq l\leq r}W_j^Q(\boldsymbol{v})=\sum\limits_{j=\eta_l,1\leq l\leq r}\max\left\{w(v_{j\mu}):1\leq \mu\leq \beta_{i}\right\}.$$
Therefore
\begin{eqnarray*}
\sum\limits_{(i,j)=(\lambda,\eta_l),\atop 1\leq l\leq r}W_{ij}^{\boldsymbol{uv}}&=&\sum\limits_{(i,j)=(\lambda,\eta_l),\atop 1\leq l\leq r}\max\left\{w(u_{i\varsigma}v_{j\mu}): 1\leq \varsigma\leq\alpha_i,1\leq \mu\leq\beta_j\right\}\\
&=&\sum\limits_{j=\eta_l,1\leq l\leq r}\max\left\{w(u_{\lambda\varsigma}v_{j\mu}): 1\leq \varsigma\leq\alpha_{\lambda},1\leq \mu\leq\beta_j\right\}\\
&\geq& \sum\limits_{j=\eta_l,1\leq l\leq r}\max\left\{w(u_{\lambda\varsigma}v_{j\mu}): 1\leq \mu\leq\beta_j\right\}\\
&=&\overline{\omega}_{w,(Q,\pi_2)}(u_{\lambda\varsigma}\boldsymbol{v}).
\end{eqnarray*}
Therefore
$$d_{w,(\mathcal {L},\pi)}(\mathcal {C})\geq d_{(Q,\pi_2)}(C_2)\left(d_{(P,\pi_1)}(C_1)-1\right)M_w+d_{w,(Q,\pi_2)}(C_2).$$

On the other hand, consider $\boldsymbol{u}\in C_1$ such that $d_{(P,\pi_1)}(\boldsymbol{u})=\lambda=d_{(P,\pi_1)}(C_1)$ and $d_{(Q,\pi_2)}(\boldsymbol{v})=r=d_{(Q,\pi_2)}(C_2)$. From above discussion, we conclude that
$$\overline{\omega}_{w,(\mathcal {L},\pi)}(\boldsymbol{u}\otimes\boldsymbol{v})\leq \lambda r M_w.$$
The result then follows.
\item If $P$ and $Q$ are antichains, then $\mathcal {L}$ is an antichain. Therefore
    $$\overline{\omega}_{w,(\mathcal {L},\pi)}(\boldsymbol{u}\otimes\boldsymbol{v})= \sum\limits_{(i,j)\in I_{\boldsymbol{u}\otimes\boldsymbol{v}}^\mathcal {L}}W_{ij}^{\boldsymbol{uv}}.$$
    The result immediately follows.
\item If $P$ and $Q$ are both chains, then $(i,j)\leq (i',j')\in\mathcal {L}$ if and only if $i\leq i'$ and $j\leq j'$. Assume that $d_{(P,\pi_1)}(\boldsymbol{u})=\lambda$ and $d_{(Q,\pi_2)}(\boldsymbol{v})=\delta$. Then
\begin{center}
\begin{equation}\label{eq1}
  \boldsymbol{u}\otimes \boldsymbol{v}=\left[
     \begin{array}{cccccc}
       G_{1,1}^{\boldsymbol{uv}}&\cdots& G_{1,\delta}^{\boldsymbol{uv}}&O&\cdots& O\\[1mm]
       \vdots&&\vdots&\vdots&&\vdots\\[1mm]
       G_{\lambda,1}^{\boldsymbol{uv}}&\cdots& G_{\lambda,\delta}^{\boldsymbol{uv}}&O&\cdots& O\\[1mm]\\
       O&\cdots& O&O&\cdots& O\\[1mm]
       \vdots&&\vdots&\vdots&&\vdots\\[1mm]
       O&\cdots& O&O&\cdots& O\\
     \end{array}
   \right]
\end{equation}
\end{center}
satisfies $G_{\lambda,\delta}^{\boldsymbol{uv}}\neq O$ is an $\alpha_{\lambda}\times\beta_{\delta}$ matrix. Then
$$\overline{\omega}_{w,(\mathcal {L},\pi)}(\boldsymbol{u}\otimes\boldsymbol{v})= W_{\lambda\delta}^{\boldsymbol{uv}}+(\lambda\delta-1)M_w.$$
Hence
$$(\lambda\delta-1)M_w+m_w\leq\overline{\omega}_{w,(\mathcal {L},\pi)}(\boldsymbol{u}\otimes\boldsymbol{v})\leq \lambda\delta M_w.$$
The result then follows.
\end{enumerate}
\end{proof}

\begin{remark}
The case for $Q$ being a chain and $P$ being an antichain is symmetric with the case $P$ being a chain and $Q$ being an antichain.
\end{remark}

\begin{remark}
Let $P$ and $Q$ are two chains and let $\pi_1$ and $\pi_2$ be labeling maps of $P$ and $Q$ respectively.
When $\pi_1(i)=1$ for all $i\in[s]$ and $\pi_2(j)=1$ for all $j\in[t]$, we have that
\begin{eqnarray*}
d_{w,(\mathcal {L},\pi)}(\mathcal {C})&=&\left(d_{(P,\pi_1)}(C_1)d_{(Q,\pi_2)}(C_2)-1\right)M_w+m_w\\
&=&\left(d_{w,(P,\pi_1)}(C_1)-m_w\right)d_{(Q,\pi_2)}(C_2)+d_{w,(Q,\pi_2)}(C_2)\\
&=&\left(d_{w,(Q,\pi_2)}(C_2)-m_w\right)d_{(P,\pi_1)}(C_1)+d_{w,(P,\pi_1)}(C_1).
\end{eqnarray*}
\end{remark}

The following corollary, which has been shown in [\ref{POMSET}], is a special case of Proposition \ref{CAR} .

\begin{corollary}
 Let $C_1$ be a linear $(P,\pi_1,w)$-code and let $C_2$ be a linear $(Q,\pi_2,w)$-code. Let $\mathcal {L}=P\otimes Q$ and $\pi=\pi_1\otimes\pi_2$. For $w$ being the Lee weight over $\mathbb{Z}_m$ where $m$ is prime (that is, $\mathbb{Z}_m$ is a field) and $\pi_i$ is trivial, we have
\begin{enumerate}[(1)]
\item if $P$ and $Q$ are antichains, then
$$d_{(P,\pi_1)}(C_1)d_{(Q,\pi_2)}(C_2)\leq d_{w,(\mathcal {L},\pi)}(\mathcal {C})\leq d_{(P,\pi_1)}(C_1)d_{(Q,\pi_2)}(C_2)\left\lfloor\frac{m}{2}\right\rfloor.$$
\item if $P$ is a chain and $Q$ is an antichain, then
$$d_{(Q,\pi_2)}(C_2)\left(d_{(P,\pi_1)}(C_1)-1\right) \left\lfloor\frac{m}{2}\right\rfloor d_{w,(Q,\pi_2)}(C_2)\leq d_{w,(\mathcal {L},\pi)}(\mathcal {C})\leq d_{(P,\pi_1)}(C_1)d_{(Q,\pi_2)}(C_2)\left\lfloor\frac{m}{2}\right\rfloor.$$
\item if $P$ and $Q$ are two chains, then
\begin{eqnarray*}
d_{w,(\mathcal {L},\pi)}(\mathcal {C})&=&\left(d_{w,(P,\pi_1)}(C_1)-1\right)d_{(Q,\pi_2)}(C_2)+d_{w,(Q,\pi_2)}(C_2)\\
&=&\left(d_{w,(Q,\pi_2)}(C_2)-1\right)d_{(P,\pi_1)}(C_1)+d_{w,(P,\pi_1)}(C_1).
\end{eqnarray*}
\end{enumerate}
\end{corollary}

\begin{proposition}\label{LEX}
Let $C_1$ be a linear $(P,\pi_1,w)$-code and let $C_2$ be a linear $(Q,\pi_2,w)$-code. Let $\mathcal {L}=P\star Q$ and $\pi=\pi_1\otimes\pi_2$. Then the following results hold:
\begin{enumerate}[(1)]
\item Suppose that $P$ is a chain such that $1<2<\cdots<s$ and $Q$ is an antichain, then
    $$\left(d_{(P,\pi_1)}(C_1)-1\right)tM_w+d_{w,(Q,\pi_2)}(C_2)\leq d_{w,(\mathcal {L},\pi)}(\mathcal {C})\leq\left(d_{(P,\pi_1)}(C_1)-1\right)tM_w+d_{(Q,\pi_2)}(C_2)M_w.$$
\item Suppose that $P$ is a chain with order relation $1<2<\cdots<s$ and $Q$ is a chain with order relation $1<2<\cdots<t$, then
     $$m_w+\left(d_{(P,\pi_1)}(C_1)-1\right)tM_w+ \left(d_{(Q,\pi_2)}(C_2)-1\right)M_w\leq d_{w,(\mathcal {L},\pi)}(\mathcal {C})\leq \left(d_{(P,\pi_1)}(C_1)-1\right)tM_w+d_{(Q,\pi_2)}(C_2)M_w .$$
\item Suppose that $P$ and $Q$ are both antichains, then
    $$d_{(P,\pi_1)}(C_1)d_{(Q,\pi_2)}(C_2) m_w\leq d_{w,(\mathcal {L},\pi)}(\mathcal {C})\leq d_{(P,\pi_1)}(C_1)d_{(Q,\pi_2)}(C_2)M_w.$$
\item Suppose that $Q$ is a chain such that $1<2<\cdots<t$ and $P$ is an antichain, then
    $$ d_{w,(P,\pi_1)}(C_1)+d_{(P,\pi_1)}(C_1)\left(d_{(Q,\pi_2)}(C_2)-1\right)M_w\leq d_{w,(\mathcal {L},\pi)}(\mathcal {C})\leq d_{(P,\pi_1)}(C_1)d_{(Q,\pi_2)}(C_2)M_w.$$
\end{enumerate}
\end{proposition}

\begin{proof}
(3) and (4) are straightforward from Proposition \ref{CAR}.
Let $\boldsymbol{u}\otimes\boldsymbol{v}\in \mathcal {C}$.
\begin{enumerate}[(1)]
\item If $P$ is a chain and $Q$ is an antichain, then $(i,j)\leq (i',j')\in\mathcal {L}$ if and only if $i<i'$ or $(i,j)=(i',j')$. Suppose that $d_{(P,\pi_1)}(\boldsymbol{u})=\lambda$ and $I_{\boldsymbol{v}}^Q=\{\eta_1,\eta_2,\ldots,\eta_r\}$.
    The similar discussion as Proposition \ref{CAR}, we have
   \begin{eqnarray*}
   \overline{\omega}_{w,(\mathcal {L},\pi)}(\boldsymbol{u}\otimes\boldsymbol{v})&=& \sum\limits_{(i,j)=(\lambda,\eta_l),\atop 1\leq l\leq r}W_{ij}^{\boldsymbol{uv}}+(\lambda-1)tM_w\\[1mm]
   &=&\sum\limits_{(i,j)=(\lambda,\eta_l),\atop 1\leq l\leq r}W_{ij}^{\boldsymbol{uv}}+\left(d_{(P,\pi_1)}(\boldsymbol{u})-1\right)tM_w\\[1mm]
   &\geq&d_{w,(Q,\pi_2)}(C_2)+ \left(d_{(P,\pi_1)}(C_1)-1\right)tM_w.
   \end{eqnarray*}
\item Suppose that $P$ and $Q$ are chains. Then $\mathcal {L}$ is a chain such that $(i,j)\leq(i',j')$ if and only if $i<i'$ or $i=i¡®$ and $j\leq j'$. Assume that $d_{(P,\pi_1)}(\boldsymbol{u})=\lambda$ and $d_{(Q,\pi_2)}(\boldsymbol{v})=\delta$. Consider the matrix (\ref{eq1}) in the proof in Proposition \ref{CAR} (3), we have
    $$\overline{\omega}_{w,(\mathcal {L},\pi)}(\boldsymbol{u}\otimes\boldsymbol{v})=W_{\lambda\delta}^{\boldsymbol{uv}} +(\lambda-1)tM_w+(\delta-1)M_w\geq m_w+\left(d_{(P,\pi_1)}(C_1)-1\right)tM_w+ \left(d_{(Q,\pi_2)}(C_2)-1\right)M_w.$$
   On the other hand, if $\lambda=d_{(P,\pi_1)}(C_1)$ and $\delta=d_{(Q,\pi_2)}(C_2)$, we conclude that
    $$\overline{\omega}_{w,(\mathcal {L},\pi)}(\boldsymbol{u}\otimes\boldsymbol{v})\leq \left(d_{(P,\pi_1)}(C_1)-1\right)tM_w+d_{(Q,\pi_2)}(C_2)M_w .$$
\end{enumerate}
\end{proof}

\begin{remark}
Let $P$ be a chain with order relation $1<2<\cdots<s$ and let $Q$ be an antichain. Let $\pi_1$ and $\pi_2$ be labeling maps of $P$ and $Q$ respectively.
When $\pi_1(i)=1$ for all $i\in[s]$ and $\pi_2(j)=1$ for all $j\in[t]$, we have that
\begin{eqnarray*}
d_{w,(\mathcal {L},\pi)}(\mathcal {C})&=&\left(d_{(P,\pi_1)}(C_1)-1\right)tM_w+d_{w,(Q,\pi_2)}(C_2)\\
&=&\left(d_{w,(P,\pi_1)}(C_1)-m_w\right)t+d_{w,(Q,\pi_2)}(C_2).
\end{eqnarray*}
\end{remark}

\begin{remark}
Let $P$ be a chain with order relation $1<\cdots<s$ and let $Q$ be a chain with order relation $1<\cdots<t$. Let $\pi_1,\pi_2$ be labeling maps of $P,Q$ respectively.
When $\pi_1(i)=1$ for all $i\in[s]$ and $\pi_2(j)=1$ for all $j\in[t]$, we have that
\begin{eqnarray*}
d_{w,(\mathcal {L},\pi)}(\mathcal {C})&=&m_w+\left(d_{(P,\pi_1)}(C_1)-1\right)tM_w+ \left(d_{(Q,\pi_2)}(C_2)-1\right)M_w\\
&=&\left(d_{w,(P,\pi_1)}(C_1)-m_w\right)t+d_{w,(Q,\pi_2)}(C_2).
\end{eqnarray*}
\end{remark}

Let $P$ be a chain with order relation $1<\cdots<s$. Set
$$D_{(P,\pi_1)}(C_1)=\max\left\{d_{(P,\pi)}(\boldsymbol{u}): \boldsymbol{u}\in C_1\right\}.$$
Similarly $D_{(Q,\pi_2)}(C_2)$ can be defined.

Denote by $R_1$ and $R_2$ the covering radius of $C_1$ and $C_2$ respectively when $w$ is taken to be Hamming weight.

\begin{remark}
Let $C_1$ be a linear $\left(P,\pi_1,w\right)$-code. When $P$ is a chain such that $1<\cdots<s$, we have that $$(R_1-1)M_w<\tilde{\rho}(C_1)\leq R_1M_w.$$
\end{remark}

\begin{theorem}\label{PRODUCT}
Let $C_1$ be a linear $(P,\pi_1,w)$-code and let $C_2$ be a linear $(Q,\pi_2,w)$-code. Suppose that $\mathcal {L}=P\otimes Q$ and $\pi=\pi_1\otimes\pi_2$. Then
    $$\tilde{\rho}(\mathcal {C})\geq \max\left\{s\tilde{\rho}(C_2),\ t\tilde{\rho}(C_1)\right\}.$$
Moreover,
\begin{enumerate}[(1)]
\item Suppose that $P$ is a chain with order relation $1<2<\cdots<s$ and $Q$ is an antichain, then
     \begin{enumerate}[(a)]
    \item $\tilde{\rho}(\mathcal {C})= stM_w$ if $D_{(P,\pi_1)}(C_1)<s$.
    \item $\tilde{\rho}(\mathcal {C})\geq R_2(s-1)M_w+R_2m_w$.
    \item $\tilde{\rho}(\mathcal {C})\leq(s-1)tM_w+\tilde{\rho}(C_2)$ if $D_{(P,\pi_1)}(C_1)=s$ and $\alpha_s=1$.
    \end{enumerate}

\item Suppose that $P$ and $Q$ are both chains with order relations $1<2<\cdots<s$ and $1<2<\cdots<t$ respectively, then
    \begin{enumerate}[(a)]
    \item
    $\tilde{\rho}(\mathcal {C})= stM_w$ if $D_{(P,\pi_1)}(C_1)< s$ or $D_{(Q,\pi_2)}(C_2)<t$.

    \item
    $\tilde{\rho}(\mathcal {C})\leq(s-1)tM_w+\tilde{\rho}(C_2)$ if $ D_{(P,\pi_1)}(C_1)=s$, $D_{(Q,\pi_2)}(C_2)=t$ and $\alpha_s=1$.
    \item
    $\tilde{\rho}(\mathcal {C})\leq(t-1)sM_w+\tilde{\rho}(C_1)$ if $ D_{(P,\pi_1)}(C_1)=s$, $D_{(Q,\pi_2)}(C_2)=t$ and $\beta_t=1$.
    \item
    $\tilde{\rho}(\mathcal {C})\leq\min\{(s-1)tM_w+\tilde{\rho}(C_2),(t-1)sM_w+\tilde{\rho}(C_1)\}$  if $ D_{(P,\pi_1)}(C_1)=s$, $D_{(Q,\pi_2)}(C_2)=t$ and $\alpha_s=\beta_t=1$.
    \item $\tilde{\rho}(\mathcal {C})\geq \max\{(sR_2-1)M_w+m_w,(tR_1-1)M_w+m_w\}$.
    \end{enumerate}

\end{enumerate}
\end{theorem}

\begin{proof}
Let $\zeta\in\mathbb{F}_q$ satisfies $w(\zeta)=M_w$. Let $\tilde{\boldsymbol{h}}\in\mathbb{F}_q^{n_2}$ be a coset leader of $C_2$ satisfing that $\overline{\omega}_{w,(Q,\pi_2)}(\tilde{\boldsymbol{h}})=\tilde{\rho}(C_2)$. Take $\boldsymbol{h}=(\tilde{\boldsymbol{h}},\ldots,\tilde{\boldsymbol{h}})^T_{1\times n_1}\in \mathbb{F}_q^{n_1n_2}$. Let $\boldsymbol{u}\otimes\boldsymbol{v}\in\mathcal {C}$. Suppose that $\boldsymbol{u}=(u_{11},\ldots,u_{1\alpha_1},\ldots,u_{s1},\ldots,u_{s\alpha_s})$. Then $$\boldsymbol{u}\otimes\boldsymbol{v}= (u_{11}\boldsymbol{v},\ldots,u_{1\alpha_1}\boldsymbol{v},\ldots, u_{s1}\boldsymbol{v},\ldots,u_{s\alpha_{s}}\boldsymbol{v})^T_{1\times n_1}$$
    and hence
\begin{eqnarray*}
d_{w,(\mathcal {L},\pi)}(\boldsymbol{h},\boldsymbol{u}\otimes\boldsymbol{v}) &=&\overline{\omega}_{w,(\mathcal {L},\pi)}(\boldsymbol{h},\boldsymbol{u}\otimes\boldsymbol{v})
\geq\sum\limits_{i=1}^{s}\overline{\omega}_{w,(Q,\pi_2)} (\tilde{\boldsymbol{h}}-u_{i\alpha_i}\boldsymbol{v})\\
&\geq&\sum\limits_{i=1}^{s}\overline{\omega}_{w,(Q,\pi_2)} (\tilde{\boldsymbol{h}})=s\tilde{\rho}(C_2).
\end{eqnarray*}
Therefore $\tilde{\rho}(\mathcal {C})\geq s\tilde{\rho}(C_2)$. In the same way. we can prove $\tilde{\rho}(\mathcal {C})\geq t\tilde{\rho}(C_1)$

\begin{enumerate}[(1)]
\item Suppose that $P$ is a chain and $Q$ is an antichain.
\begin{itemize}

\item {\bf (a): }
Suppose that $D_{(P,\pi_1)}(C_1)< s$, then $\boldsymbol{g}\in \mathcal {C}$ has the form
    \begin{equation}\label{eq3}
    \left[\begin{array}{ccc}
       G_{1,1}&\cdots&G_{1,t}\\[1mm]
       \vdots&&\vdots\\[1mm]
       G_{s-1,1}&\cdots&G_{s-1,t}\\[1mm]
       O&\cdots&O\\
           \end{array}
   \right]
   \end{equation}
   where $G_{i,j}$ is a $\alpha_i\times\beta_j$ matrix. Take $\boldsymbol{h}\in\mathbb{F}_q^{n_1n_2}$ of the form
   \begin{equation}\label{eq2}
    \left[\begin{array}{ccc}
       H_{1,1}&\cdots&H_{1,t}\\[1mm]
       \vdots&&\vdots\\[1mm]
       H_{s-1,1}&\cdots&H_{s-1,t}\\[1mm]
       H_{s,1}&\cdots&H_{s,t}\\
      \end{array}
   \right]
   \end{equation}
   where $H_{s,r}$ is a $\alpha_s\times \beta_r$ such that $\zeta$ is an element of $H_{s,r}$ for $1\leq r\leq t$. Then for any $\boldsymbol{c}\in\mathcal {C}$, we have $\overline{\omega}_{w,(\mathcal {L},\pi)}(\boldsymbol{h}-\boldsymbol{c})=stM_w$ and hence $\min\limits_{\boldsymbol{c}\in C}d_{w,(\mathcal {L},\pi)}(\boldsymbol{h},\boldsymbol{c})=stM_w$ which implies that $\tilde{\rho}(\mathcal {C})=stM_w$.

\item {\bf (b): } Let $\boldsymbol{h}\in\mathbb{F}_q^{n_1n_2}$ have the form (\ref{eq2}). Write
\begin{equation}\label{eq4}
 [H_{s,1}\cdots H_{s,t}]=\left[\begin{array}{ccc}
       h_{11}&\cdots&h_{1n_2}\\[1mm]
       \vdots&&\vdots\\[1mm]
       h_{\alpha_s1}&\cdots&h_{\alpha_sn_2}\\
      \end{array}
   \right].
 \end{equation}
  Suppose that $\boldsymbol{h_1}=(h_{11},\ldots,h_{1n_2})\in\mathbb{F}_q^{n_2}$ such that $\min\limits_{\boldsymbol{c}\in C_2}d_{(Q,\pi_2)}(\boldsymbol{h_1},\boldsymbol{c})=R_2$. Then
  $$\overline{\omega}_{w,(\mathcal {L},\pi)}(\boldsymbol{h}- \boldsymbol{u}\otimes\boldsymbol{v})\geq (s-1)R_2M_w+R_2m_w.$$

\item {\bf (c): }
Suppose that $D_{(P,\pi_1)}(C_1)=s$ and $\alpha_s=1$. Let $\boldsymbol{h}\in\mathbb{F}_q^{n_1n_2}$.
\begin{itemize}
  \item  If $\overline{\omega}_{w,(\mathcal {L},\pi)}(\boldsymbol{h})\leq (s-1)tM_w$, then $\min\limits_{\boldsymbol{c}\in C}d_{w,(\mathcal {L},\pi)}(\boldsymbol{v},\boldsymbol{c})\leq (s-1)tM_w$ since $\boldsymbol{0}\in\mathcal {C}$.
  \item If $\overline{\omega}_{w,(\mathcal {L},\pi)}(\boldsymbol{h})>(s-1)tM_w$, then $\boldsymbol{v}$ has the form (\ref{eq2}) such that not all $H_{sr}$ are $O$ for $1\leq r\leq t$. Write
      $$[H_{s1}\cdots H_{st}]=(h_1,\ldots,h_{\beta_1},\ldots, h_{\beta_1+\cdots+\beta_{t-1}+1},\ldots,h_{n_2})= \boldsymbol{h_s}\in\mathbb{F}_q^{n_2}.$$
For $\boldsymbol{h_s}\in\mathbb{F}_q^{n_2}$, there exists $\boldsymbol{v}\in C_2$ such that $d_{w,(Q,\pi_2)}(\boldsymbol{h_s},\boldsymbol{v})\leq\tilde{\rho}(C_2)$. Consider $\boldsymbol{0}\neq \boldsymbol{u}\otimes\boldsymbol{v}\in\mathcal {C}$ with $\boldsymbol{u}=(u_1,\ldots,u_{\alpha_1},\ldots,u_{\alpha_1+\cdots+\alpha_{s-1}},1)\in C_1$, then $$\boldsymbol{u}\otimes\boldsymbol{v}= (u_{1}\boldsymbol{v},\ldots,u_{\alpha_1}\boldsymbol{v},\ldots, u_{\alpha_1+\cdots+\alpha_{s-1}}\boldsymbol{v},\boldsymbol{v})^{T}.$$
Therefore
$$\overline{\omega}_{w,(\mathcal {L},\pi)}(\boldsymbol{h}-\boldsymbol{u}\otimes\boldsymbol{v})\leq(s-1)tM_w+\tilde{\rho}(C_2).$$
\end{itemize}
To sum up, we conclude that
$$\tilde{\rho}(\mathcal {C})\leq(s-1)tM_w+\tilde{\rho}(C_2).$$
\end{itemize}

\item The proof is similar to (1) and hence we omit it.

\end{enumerate}
\end{proof}

\begin{theorem}
Let $C_1$ be a linear $(P,\pi_1,w)$-code and let $C_2$ be a linear $(Q,\pi_2,w)$-code. Suppose that $\mathcal {L}=P\star Q$ and $\pi=\pi_1\otimes\pi_2$. Then
$$\tilde{\rho}(\mathcal {C})\geq \max\left\{s\tilde{\rho}(C_2),\ t\tilde{\rho}(C_1)\right\}.$$
Moreover,
\begin{enumerate}[(1)]
\item If $P$ is a chain such that $1<2<\cdots<s$, then
    \begin{enumerate}[(a)]
    \item
    $\tilde{\rho}(\mathcal {C})= stM_w$ if $D_{(P,\pi_1)}(C_1)< s$.\\
    Especially when $Q$ is a chain,  we have $\tilde{\rho}(\mathcal {C})= stM_w$ if $D_{(P,\pi_1)}(C_1)< s$ or $D_{(Q,\pi_2)}(C_2)< t$.
    \item
    $\tilde{\rho}(\mathcal {C})\leq(s-1)tM_w+\tilde{\rho}(C_2)$ if $ D_{(P,\pi_1)}(C_1)=s$, $D_{(Q,\pi_2)}(C_2)=t$ and $\alpha_s=1$.
    \item
    $\tilde{\rho}(\mathcal {C})\geq \max\left\{(s-1)tM_w+\tilde{\rho}(C_2),\ t\tilde{\rho}(C_1)\right\}.$
        \end{enumerate}
\item If $P$ is an antichain and $Q$ is a chain with order relation $1<2<\cdots<t$, then
     \begin{enumerate}[(a)]
    \item $\tilde{\rho}(\mathcal {C})= stM_w$ if $D_{(Q,\pi_2)}(C_2)<t$.
    \item $\tilde{\rho}(\mathcal {C})\leq(t-1)sM_w+\tilde{\rho}(C_1)$ if $D_{(Q,\pi_2)}(C_2)=t$ and $\beta_t=1$.
    \end{enumerate}
 \end{enumerate}
\end{theorem}

\begin{proof}

Let $P$ be a chain and let $Q$ be an antichain.
Let $\tilde{\boldsymbol{h}}\in\mathbb{F}_q^{n_2}$ be a coset leader of $C_2$ satisfing that $\overline{\omega}_{w,(Q,\pi_2)}(\tilde{\boldsymbol{h}})=\tilde{\rho}(C_2)$.
Take $\boldsymbol{h}=(\tilde{\boldsymbol{h}},\ldots,\tilde{\boldsymbol{h}})^T_{1\times n_1}$. Then $d_{w,(\mathcal {L},\pi)}(\boldsymbol{h},\boldsymbol{0})=(s-1)tM_w+\tilde{\rho}(C_2)$. Let $\boldsymbol{0}\neq\boldsymbol{u}\otimes\boldsymbol{v}\in\mathcal {C}$. Suppose that $\overline{\omega}_{(P,\pi_1)}(\boldsymbol{u})=\lambda$. Then
$$d_{w,(\mathcal {L},\pi)}(\boldsymbol{h},\boldsymbol{u}\otimes\boldsymbol{v}) \geq(s-1)tM_w+\tilde{\rho}(C_2).$$
The rest of the proof is on similar lines to Theorem \ref{PRODUCT} and hence we omit it.
\end{proof}

\section{Conclusion}

Now we show that our results lead to several previous results.
\begin{enumerate}[(1)]
\item When all blocks are trivial and $w$ is taken to be the Lee weight over $\mathbb{F}_q$, the results coincide with the results under pomset metric as seen in  [\ref{POMSET}].
\item When $w$ is taken to be the Hamming weight over $\mathbb{F}_q$, our results coincide with the poset block metric case. In particular, when all block are trivial, the results coincide with the poset metric case.
\item In case both conditions occur (all blocks are trivial, $w$ is the Hamming weight and $P$ is the antichain order), the results coincide with the result under classical Hamming metric as seen in [\ref{HUFFMAN}].
\end{enumerate}

\end{document}